\documentclass[journal,onecolumn,Manuscript]{IEEEtran}

\ifCLASSINFOpdf
  
\else
  
\fi

\usepackage{amsmath}
\usepackage{amsthm}
\usepackage{cases}
\usepackage{amsfonts}
\usepackage{cite}
\usepackage{amssymb}
\usepackage{enumerate}
\usepackage{graphicx}
\usepackage{float}
\usepackage{caption}
\usepackage{verbatim}
\usepackage{dblfloatfix}
\usepackage[cmintegrals]{newtxmath}
\usepackage[ruled,linesnumbered]{algorithm2e}
\usepackage{accents}
\newcommand{\ubar}[1]{\underaccent{\bar}{#1}}

\newtheorem{theorem}{Theorem}
\newtheorem{lemma}{Lemma}
\newtheorem{definition}{Definition}
\newtheorem{corollary}{Corollary}
\newtheorem{remark}{Remark}

\begin{document}

\title{On General Capacity-Distortion Formulas of Integrated Sensing and Communication}

\author{
  \IEEEauthorblockN{Yiqi Chen\IEEEauthorrefmark{1},
  Tobias Oechtering\IEEEauthorrefmark{2},
  Mikael Skoglund\IEEEauthorrefmark{2},
  and Yuan Luo\IEEEauthorrefmark{1}}\\
\IEEEauthorblockA{\IEEEauthorrefmark{1}%
Shanghai Jiao Tong University,  
Shanghai,
China, 
\{chenyiqi,yuanluo\}@sjtu.edu.cn}\\
\IEEEauthorblockA{\IEEEauthorrefmark{2}%
KTH Royal Institute of Technology,  
100 44 Stockholm,
Sweden, 
\{oech,skoglund\}@kth.se}
}


\maketitle
\thispagestyle{empty}
\pagestyle{empty}
\begin{abstract}
The integrated sensing and communication (ISAC) problem with general state and channel distributions is investigated. General formulas of the capacity-distortion tradeoff for the ISAC problem under maximal and average distortion constraints are provided. The results cover some existing communication models such as the general point-to-point channel and Gel'fand-Pinsker channel. More details including memoryless states and channels, mixed states and channels, and rate-limited CSI at one side are considered. Numerical results focus on writing-on-dirty paper model and channel with ergodic/non-ergodic fading.
\end{abstract}


% \fi
\IEEEpeerreviewmaketitle

\section{introduction}
One attractive feature of the next-generation wireless communication systems is the ability to autonomously learn and adapt to the environment, which enables the participants in the system to react to changes. Such an intelligent behavior of the system relies on the ability to sense/estimate the environment, which is characterized by `states' from the point of view of information theory. In the integrated sensing and communication (ISAC) problem, on top of the reliable communication between the sender and receiver over a state-dependent channel, one of the participants is required to estimate the channel state.

In \cite{sutivong2005channel} and \cite{zhang2011joint}, ISAC problems where the estimation constraint is imposed on the receiver side were investigated. In \cite{sutivong2005channel}, the transmitter has state information knowledge and wants to reveal it to the receiver through communication. Minimal achievable distortion and the capacity-distortion function were given in the paper. In \cite{zhang2011joint}, a similar model was considered while neither the transmitter nor the receiver has state information knowledge, except its statistics. The ISAC with causal CSI at the transmitter and distortion constraint at the receiver was studied in \cite{choudhuri2010capacity}.

The information-theoretic analysis of ISAC problem where both transmitting and estimation are performed at the transmitter side was initiated by \cite{ahmadipour2022information}, where point-to-point channel, multiple access channel, and broadcast channels were considered. The authors gave the capacity-distortion tradeoff and also used a Blahut-Arimoto algorithm to evaluate the results numerically. In their setting, the transmitter tries to transmit a pure message to the receiver, and then receives feedback from the receiver. The transmitter uses all the resources (e.g. input, feedback) to estimate the channel states. Secure ISAC was studied in \cite{gunlu2023secure}, where the transmitter uses feedback to both perform estimation and enhance the secrecy transmission. Fundamental limits of ISAC over Gaussian channels were provided in \cite{xiong2023fundamental}. However, the model in \cite{ahmadipour2022information} and its following works make an assumption that the states are i.i.d. generated by a fixed distribution, which is not usually the case in real-world communication systems. Hence, investigating the ISAC with a more general state/channel setting is the topic of this paper.

Another line of works of state-dependent channels is the action-dependent channel, which is first studied in \cite{weissman2010capacity}. The model was then extended to multi-user case\cite{dikstein2014mac,steinberg2012degraded,steinberg2013degraded}. Action-dependent channel with side information and reconstruction requirement was investigated in \cite{kittichokechai2015coding}. Secrecy problems of action-dependent communication were studied in \cite{kittichokechai2015secure} for source coding problems and \cite{dai2013wiretap}\cite{dai2020impact} for channel coding problems. 

Coding for general source/channel problems was discussed in \cite{verdu1994general,han1993approximation,han2006information,koga2013information}, where the distributions of the source and channel can be arbitrary, and so are the alphabets of input and output symbols. The general wiretap channel coding problem was studied in \cite{bloch2008secrecy}. Gel'fand-Pinsker coding was extended to general state and general channel case in \cite{tan2014formula}. The author further gave the capacity results for coded side information at one side and full information at another side, and the case for mixed state distribution and mixed channel. The Wyner-Ziv coding for general sources was discussed in \cite{iwata2002information} under the maximal distortion criterion, and then studied in \cite{yang2007wyner} under the average distortion criterion.

In this paper, we consider the ISAC problem with a general action-dependent state and channel setting, with different noisy side information available at the encoder and decoder sides. The distributions of state and channel may be arbitrarily nonstationary and/or arbitrarily nonergodic with abstract input, output and state alphabets. Our problem arises in real-world communication for instance by considering a base station monitoring and controlling a vehicle. The vehicle moves according to the instructions from the base station, each corresponding to a fixed route. The selection of the action determines the route that the vehicle is going to go through, and the base station tries to estimate some states or properties related to the vehicle. We assume the noisy side information at the encoder and decoder sides since once the route is determined, the base station may have a prior estimation of what they are interested in, and the vehicle itself can be equipped with some sensors to detect the data that is related to the states. Capacity-distortion tradeoff results for both maximal distortion and average distortion are provided. We further investigate the case for mixed states and mixed channels, and channels with rate-limited side information. Channel with rate-limited side information was first studied in \cite{heegard1983capacity} and an inner bound of the capacity was given. We extended the results to general cases in this paper. We also give some numerical results for ergodic/non-ergodic state sources by considering some special cases of the model including writing-on-dirty-paper channel and fading channel.

The rest of the paper is organized as follows. In Section \ref{sec: definitions} we provide the notations and definitions used in this paper. Section \ref{sec: main results} presents the results of this paper. We give our main theorems at the beginning of Section \ref{sec: main results} and provide the results for the memoryless case, mixed state/channel case and rate-limited CSI case at Sections \ref{sec: memoryless}, \ref{sec: mixed case} and \ref{sec: rate-limited csi}, respectively. Sections \ref{sec: proof of average distortion} and \ref{sec: proof of maximal distortion} prove Theorems \ref{the: average distortion capacity} and \ref{the: maximal distortion capacity}, respectively. In Section \ref{sec: numerical examples}, we show numerical examples by applying our results to the writing-on-dirty paper model and fading channel model.

\section{models and definitions}\label{sec: definitions}
\subsection{Notations}
Throughout this paper, random variables and sample values are denoted by capital letters and lowercase letters, e.g. $X$ and $x$. Sets are denoted by calligraphic letters. We use $\boldsymbol{X}=\{X^n=(X^{(n)}_1,X^{(n)}_2,\dots,X^{(n)}_n)\}_{n=1}^{\infty}$ to denote a general source, where $X^n$ represents an $n-$length random sequence. If the random sequence has i.i.d. components, the notation can be simplified to $X^n=(X_1,X_2,\dots,X_n)$.  The n-length sample sequence is written by $x^n=(x_1,x_2,\dots,x_n)$. Let $\mathcal{X}^n$  be the n-fold Cartesian product of $\mathcal{X}$, which is the set of all possible $x^n$. To denote substrings, let $X^{i}=(X_1,X_2,\dots,X_i)$ and $X^{n}_{i+1}=(X_{i+1},X_{i+2},\dots,X_n)$.  We use $P_X$ to denote the probability distribution of a random variable $X$ and $P_{XY},P_{X|Y}$ to denote the joint distribution and conditional distribution, respectively. The corresponding n-length general distributions are $P_{X^n},P_{X^nY^n}$, and $P_{X^n|Y^n}$, respectively. For a sequence $\boldsymbol{x}$ generated i.i.d. according to some distribution $P_X$, we denote $P_X^n(\boldsymbol{x})=\prod_{i=1}^n P_X(x_i)$. For given random variable $X$ with distribution $P_X$, we follow the convention in \cite{tan2014formula} and use  $\mathbb{E}[X]=\sum_{x\in\mathcal{X}}P_X(x)\cdot x$ to represent its expectation, even though the alphabet $\mathcal{X}$ may not be countable and $\int_{x} x dP_{X}$ would be more precise. For simplicity, we use $\sum_x$ to represent the sum over all elements of the set where $x$ is defined.

Next, we introduce the terms of \emph{limsup} and \emph{liminf} \emph{in probability} as done by Han and Verd\'u\cite{koga2013information}.
\begin{definition}
  Let $\{Z_n\}_{n=1}^{\infty}$ be a sequence of real-valued random variables. The limsup in probability of $Z_n$ is
  \begin{align*}
    \emph{p}-\limsup_{n\to\infty} Z_n := \emph{inf}\{\alpha | \lim_{n\to\infty}\emph{Pr}\{Z_n > \alpha\} = 0\}.
  \end{align*}
  The liminf in probability of $Z_n$ is
  \begin{align*}
    \emph{p}-\liminf_{n\to\infty} Z_n := \emph{sup}\{\beta | \lim_{n\to\infty}\emph{Pr}\{Z_n < \beta\} = 0\}.
  \end{align*}
\end{definition}
The following definitions from \cite{koga2013information} play an important role in characterizing the capacity results in this paper.
\begin{definition}
  Given a pair of stochastic processes $(\boldsymbol{X},\boldsymbol{Y})=\{X^n,Y^n\}_{n=1}^{\infty}$ with joint distribution $\{P_{X^nY^n}\}_{n=1}^{\infty}$, the spectral inf-mutual information rate is 
  \begin{align*}
    \ubar{I}(\boldsymbol{X};\boldsymbol{Y}) := \emph{p}-\liminf_{n\to\infty}\frac{1}{n}\log \frac{P_{Y^n|X^n}(Y^n|X^n)}{P_{Y^n}(Y^n)}.
  \end{align*}
  The spectral sup-mutual information rate is
  \begin{align*}
    \bar{I}(\boldsymbol{X};\boldsymbol{Y}) := \emph{p}-\limsup_{n\to\infty}\frac{1}{n}\log \frac{P_{Y^n|X^n}(Y^n|X^n)}{P_{Y^n}(Y^n)}.
  \end{align*}
  The spectral inf- and sup- conditional mutual information rates are defined similarly. 
\end{definition}
\subsection{Models}
\begin{figure}
  \centering
  \includegraphics[scale=0.6]{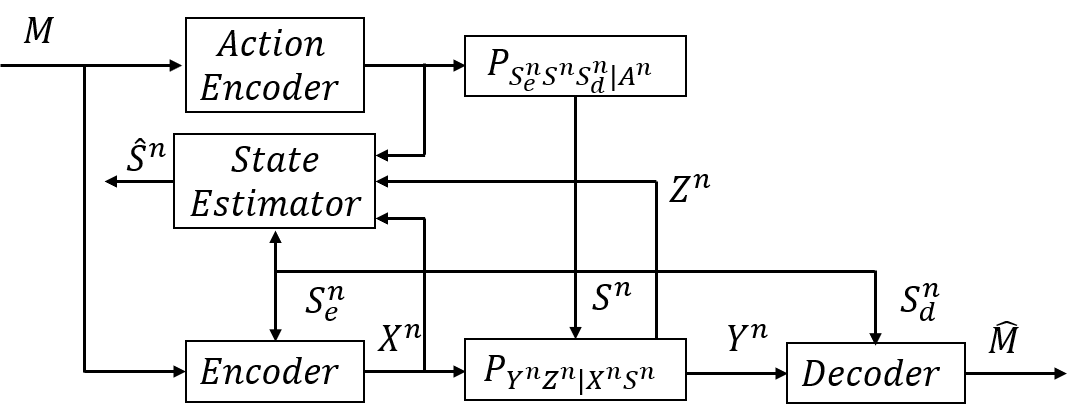}
  \caption[short]{Action-dependent Integrated Sensing and Communication Model: The state sequence and its noisy versions are generated according to the action sequence, which is determined by the message selected by the sender. The state estimator uses the channel feedback, noisy channel state information and input sequence to reproduce the state sequence according to a reconstruction function. The receiver decodes the message based on the channel output and the decoder side noisy channel state information.}
  \label{fig: channel model}
\end{figure}
Consider the case where the sender wants to send a message through a general channel $\boldsymbol{W}=\{P_{Y^nZ^n|X^nS^n}\}_{n=1}^{\infty}$, which is defined by a sequence of state-dependent channels with input alphabet $\mathcal{X}$, state alphabet $\mathcal{S}$ and output alphabets $\mathcal{Y},\mathcal{Z}$. The alphabets are not necessarily finite, and the channels are not necessarily stationary and memoryless. For every $n$, the channel is action-dependent, which means given a message $M$, uniformly distributed on message set $\mathcal{M}_n$, an action sequence $A^n$ is selected. The channel state sequence $S^n$ and its two lossy versions $S^n_e,S^n_d$ are then generated according to distribution $P_{S^n_eS^nS^n_d|A^n}$. The imperfect channel state information $S^n_e$ and $S^n_d$ are revealed to the sender and receiver, respectively. The goal of the decoder is to decode the message on the basis of its channel output observation $Y^n$ and imperfect CSI $S^n_d$. Another channel output $Z^n$, is used as feedback and sent back to the state estimator. The job of the state estimator is to reconstruct the state sequence of the channel using all the information it has, including feedback $Z^n$, imperfect CSI at the encoder side $S^n_e$, the action sequence $A^n(M)$ and the input sequence $X^n$ according to a reconstruction function $g_n$. The performance of the state estimator is evaluated by a given distortion function $d_n: \mathcal{S}^n\times\hat{\mathcal{S}}^n\to [0,+\infty)$. The channel model is depicted in Fig. \ref{fig: channel model}. Such a model is named an action-dependent integrated sensing and communication (ISAC) model.

Different from the model in \cite{ahmadipour2022information} with stationary and memoryless channel and additive distortion function assumption, the model considered here is a general channel and the distortion function $d_n,n=1,2,\dots$ can be more general. Hence, the capacity-distortion function is considered in two different criteria: average distortion criterion and maximal distortion criterion.
\begin{definition}\label{def: code}
  An $(n,\mathcal{M}_n,\mathcal{A}_n)$ code for general action-dependent integrated sensing and communication model under average distortion constraint consists of:
  \begin{itemize}
    \item A message set $\mathcal{M}_n$;
    \item An action set $\mathcal{A}_n$;
    \item An action encoder $f_A: \mathcal{M}_n \to \mathcal{A}_n$;
    \item A message encoder $f: \mathcal{M}_n \times \mathcal{S}^n_e \to \mathcal{X}^n$;
    \item A decoder $h:\mathcal{Y}^n \times\mathcal{S}^n_d \to \mathcal{M}_n;$
    \item A state estimator $g_n: \mathcal{M}_n \times \mathcal{S}^n_e \times \mathcal{Z}^n \to \hat{\mathcal{S}}^n$
  \end{itemize}
  such that its average decoding error is defined by
  \begin{align*}
    P_{e,n} = \frac{1}{|\mathcal{M}_n|}\sum_{m}\sum_{s^n,s^n_e,s^n_d}P_{S^n_eS^nS^n_d|A^n}(s^n_e,s^n,s^n_d|f_A(m))\sum_{z^n}\sum_{\substack{y^n:\\(y^n,s^n_d)\notin h^{-1}(m)}}W^n(y^n,z^n|f(m,s^n_e),s^n)
  \end{align*}
  where $h^{-1}(m)=\{(y^n,s^n_d):h(y^n,s^n_d)=m\}$ and average distortion
  \begin{align*}
    \limsup_{n\to\infty}\frac{1}{n}\mathbb{E}\left[ d_n(S^n,g_n(f(M,S^n_e),f_A(M),S^n_e,Z^n)) \right] ,
  \end{align*}
  where $d_n:\mathcal{S}^n\times\hat{\mathcal{S}}^n \to [0,+\infty)$ is the distortion function. For simplicity, sometimes we write the average distortion as $D(\boldsymbol{S},\hat{\boldsymbol{S}})$.
\end{definition}
\begin{definition}
  An $(n,\mathcal{M}_n,\mathcal{A}_n)$ code for general action-dependent integrated sensing and communication model under maximal distortion constraint has a message set, an action set, an action encoder, a message encoder, a decoder, a state estimator and average decoding error as defined in Definition \ref{def: code} with maximal distortion constraint
  \begin{align*}
    p-\limsup_{n\to \infty}\frac{1}{n}d_n(S^n,g_n(f(M,S^n_e),f_A(M),S^n_e,Z^n)).
  \end{align*}
  For simplicity, sometimes we write the maximal distortion as $\bar{D}(\boldsymbol{S},\hat{\boldsymbol{S}})$.
\end{definition}
Now we give the definitions of the achievable rate-distortion pairs regarding the above definitions of code.
\begin{definition}[$m-$achievability]
  A rate-distortion pair $(R,D)$ is said to be maximal-achievable ($m-$achievable) if for every $\epsilon>0$, there exists a sequence of $(n,\mathcal{M}_n,\mathcal{A}_n)$ codes for sufficiently large $n$ such that
  \begin{equation}\label{def: fm achievable condition}
    \begin{split}
      \liminf_{n\to\infty}\frac{1}{n}\log |\mathcal{M}_n| \geq R,\\
    \limsup_{n\to\infty}P_{e,n} \leq \epsilon,\\
    p-\limsup_{n\to \infty}\frac{1}{n}d_n(S^n,g_n(f(M,S^n_e),f_A(M),S^n_e,Z^n)) \leq D.
    \end{split}
  \end{equation}
  The $m-$capacity-distortion region is denoted by $C_m(D)$.
\end{definition}

\begin{definition}[$a-$achievability]
  A rate-distortion pair $(R,D)$ is said to be average-achievable ($a-$achievable) if for every $\epsilon>0$, there exists a sequence of $(n,\mathcal{M}_n,\mathcal{A}_n)$ codes for sufficiently large $n$ such that
  \begin{equation}\label{def: fa achievable condition}
    \begin{split}
      \liminf_{n\to\infty}\frac{1}{n}\log |\mathcal{M}_n| \geq R,\\
    \limsup_{n\to\infty}P_{e,n} \leq \epsilon,\\
    \limsup_{n\to\infty}\frac{1}{n}\mathbb{E}\left[ d_n(S^n,g_n(f(M,S^n_e),f_A(M),S^n_e,Z^n)) \right] \leq D.
    \end{split}
  \end{equation}
  The $a-$capacity-distortion region is denoted by $C_a(D)$.
\end{definition}

Intuitively, given the same distortion constraint $D$, the capacity of action-dependent ISAC under average distortion constraint should be no smaller than the capacity under maximal distortion constraint, since the latter one is a more strict condition. %It is proved in Corollary \ref{coro: average capacity larger than maximal capacity}, given any code for the action-dependent ISAC under maximal distortion constraint $D$, we can always construct a new code for the same channel model under average distortion constraint $D$. 

Now consider the case that the communication resource to one of the encoder and decoder is limited. Take limited CSI at the encoder side as an example. In this case, only state sequence $S^n$ and imperfect CSI $S^n_d$ are generated. A lossy description of $S^n_d$ is then generated and sent to the encoder. A more general model for limited CSI and stationary and memoryless channel is discussed in \cite{heegard1983capacity}. Rate-limited CSI at only the encoder or decoder side was further investigated in \cite{rosenzweig2005channels} and \cite{steinberg2008coding}. \cite{tan2014formula} generalized the previous results to general source and general channel.
\begin{definition}\label{def: rate-limited code}
  An $(n,\mathcal{M}_n,\mathcal{A}_n,\mathcal{M}_{e,n})$ code for general action-dependent integrated sensing and communication model with rate-limited CSI at encoder under average/maximal distortion constraints consists of:
  \begin{itemize}
    \item A message set $\mathcal{M}_n$;
    \item An action set $\mathcal{A}_n$;
    \item An action encoder $f_A: \mathcal{M}_n \to \mathcal{A}_n$;
    \item A side information encoder $f_e: \mathcal{S}^n_d \to \mathcal{M}_{e,n}$;
    \item A message encoder $f: \mathcal{M}_n \times \mathcal{M}_{e,n} \to \mathcal{X}^n$;
    \item A decoder $h:\mathcal{Y}^n \times\mathcal{S}^n_d \to \mathcal{M}_n;$
    \item A state estimator $g_n: \mathcal{M}_n \times \mathcal{M}_{e,n} \times \mathcal{Z}^n \to \hat{\mathcal{S}}^n$
  \end{itemize}
  such that its average decoding error is defined by
  \begin{align*}
    P_{e,n} = \frac{1}{|\mathcal{M}_n|}\sum_{m}\sum_{s^n,s^n_d}P_{S^nS^n_d|A^n}(s^n,s^n_d|f_A(m))\sum_{z^n}\sum_{\substack{y^n:\\(y^n,s^n_d)\notin h^{-1}(m)}}W^n(y^n,z^n|f(m,f_e(s^n_d)),s^n),
  \end{align*}
  where $h^{-1}(m)=\{(y^n,s^n_d):h(y^n,s^n_d)=m\}$ and average distortion
  \begin{align*}
    \limsup_{n\to\infty}\frac{1}{n}\mathbb{E}\left[ d_n(S^n,g_n(f(M,f_e(S^n_d)),f_A(M),Z^n)) \right] ,
  \end{align*}
   or maximal distortion
  \begin{align*}
    p-\limsup_{n\to\infty} \frac{1}{n}d_n(S^n,g_n(f(M,f_e(S^n_d)),f_A(M),Z^n)),
  \end{align*}
  where $d_n:\mathcal{S}^n\times\hat{\mathcal{S}}^n \to [0,+\infty)$ is the distortion function
\end{definition}

\begin{definition}
  A rate distortion pair $(R,R_e,D)$ is said to be $m-$achievable if for every $\epsilon>0$, there exists a sequence of  $(n,\mathcal{M}_n,\mathcal{A}_n,\mathcal{M}_{e,n})$ codes for sufficiently large $n$ such that in addition to \eqref{def: fm achievable condition}, it further holds that
  \begin{align}
    \label{def: rate-limited CSI at encoder achievable condition}\limsup_{n\to\infty}\frac{1}{n}\log |\mathcal{M}_{e,n}| \leq R_e.
  \end{align}
  The set of all achievable rate-distortion pairs is the capacity region $C_{e,m}(D)$. The $a-$achievable rate follows similarly except condition \eqref{def: fa achievable condition} needs to be satisfied instead of \eqref{def: fm achievable condition}. The capacity region is denoted by $C_{e,a}(D)$.
\end{definition}
The following definition characterizes a set of distortion functions with a specific property, which is useful for problems with average distortion constraints\cite{yang2007wyner}.
\begin{definition}[Uniformly bounded distortion function]
  A set of distortion functions $\{d_n\}_{n=1}^{\infty}$ is uniformly bounded if there exists $D_{max}>0$ such that
  \begin{align*}
    0\leq \frac{1}{n}d_n(S^n,\hat{S}^n) \leq D_{max}, \;\;\forall n\geq 1,\forall (S^n,\hat{S}^n)\in\mathcal{S}^n\times\mathcal{\hat{S}}^n.
  \end{align*}
\end{definition}
\begin{remark}
  In \cite{koga2013information}, for lossy source coding under average distortion constraint, the distortion function is required to satisfy the uniform integrability\cite[Chap. 5]{koga2013information}, which is a more general assumption compared to a uniformly bounded assumption. By \cite[Eq. (5.3.4)]{koga2013information}, for distortion functions satisfying uniform integrability with the source, one can always find a reference word $r^n$ such that the reconstruction distortion can be upper bounded. However, in the source coding problem the distortion is determined once the encoder selects the lossy description while in our problem, the reconstruction function relies on the channel feedback $Z^n$ and the input sequence $X^n$, which are generated by some conditional distributions based on the transmitted message. This uncertainty requires a stronger assumption on the distortion function.
\end{remark}
\section{main results}\label{sec: main results}
In this section, we present our capacity-distortion results. The following two theorems are the main results of this paper, which are capacity-distortion results for the average distortion criterion and maximal distortion criterion, respectively. We then consider three special cases of the model including discrete memoryless case, mixed states and channels case, and rate-limited cases, and provide the capacity results in Sections \ref{sec: memoryless}, \ref{sec: mixed case}, and \ref{sec: rate-limited csi}, respectively.
\begin{theorem}\label{the: average distortion capacity}
  The $a-$capacity of action-dependent integrated sensing and communication channel under average distortion constraint and uniformly bounded distortion function is
  \begin{align*}
    C_{a}(D) &= \sup_{\mathcal{P}_D} \ubar{I}(\boldsymbol{A},\boldsymbol{U};\boldsymbol{Y},\boldsymbol{S}_d) - \bar{I}(\boldsymbol{U};\boldsymbol{S}_e|\boldsymbol{A})\\
    &=\sup_{\mathcal{P}_D} \ubar{I}(\boldsymbol{U};\boldsymbol{Y},\boldsymbol{S}_d) - \bar{I}(\boldsymbol{U};\boldsymbol{S}_e|\boldsymbol{A})
  \end{align*}
  where $\mathcal{P}_D$ is the set of random processes in which each collection of random variables $(A^n,U^n,S^n_e,S^n,S^n_d,X^n,Y^n,Z^n)$ is distributed as $P_{A^n}P_{S^n_eS^nS^n_d|A^n}P_{U^n|A^nS^n_e}P_{X^n|U^nS^n_e}P_{Y^nZ^n|X^nS^n}$ and $\limsup_{n\to\infty}\frac{1}{n}\mathbb{E}\left[ d_n(S^n,g_n(X^n,A^n,S^n_e,Z^n)) \right] \leq D$.
\end{theorem}

\begin{theorem}\label{the: maximal distortion capacity}
  The $m-$capacity of action-dependent integrated sensing and communication channel under maximal distortion constraint is
  \begin{align*}
    C_{m}(D) &= \sup_{\bar{\mathcal{P}}_D} \ubar{I}(\boldsymbol{A},\boldsymbol{U};\boldsymbol{Y},\boldsymbol{S}_d) - \bar{I}(\boldsymbol{U};\boldsymbol{S}_e|\boldsymbol{A})\\
    &=\sup_{\bar{\mathcal{P}}_D} \ubar{I}(\boldsymbol{U};\boldsymbol{Y},\boldsymbol{S}_d) - \bar{I}(\boldsymbol{U};\boldsymbol{S}_e|\boldsymbol{A})
  \end{align*}
  where $\bar{\mathcal{P}}_D$ is the set of random processes in which each collection of random variables $(A^n,U^n,S^n_e,S^n,S^n_d,X^n,Y^n,Z^n)$ is distributed as $P_{A^n}P_{S^n_eS^nS^n_d|A^n}P_{U^n|A^nS^n_e}P_{X^n|U^nS^n_e}P_{Y^nZ^n|X^nS^n}$ and $p-\limsup_{n\to\infty}\frac{1}{n}d_n(S^n,g_n(X^n,A^n,S^n_e,Z^n)) \leq D$.
\end{theorem}
\begin{remark}\label{rem: equivalence of general capacity expressions}
  Here we prove the equivalent expressions in Theorem \ref{the: average distortion capacity} and Theorem \ref{the: maximal distortion capacity}. It is obvious that
  \begin{align*}
    \ubar{I}(\boldsymbol{A},\boldsymbol{U};\boldsymbol{Y},\boldsymbol{S}_d) - \bar{I}(\boldsymbol{U};\boldsymbol{S}_e|\boldsymbol{A})\geq \ubar{I}(\boldsymbol{U};\boldsymbol{Y},\boldsymbol{S}_d) - \bar{I}(\boldsymbol{U};\boldsymbol{S}_e|\boldsymbol{A})
  \end{align*}
for both distortion constraints. For the opposite direction, note that
\begin{align}
  \ubar{I}(\boldsymbol{A},\boldsymbol{U};\boldsymbol{Y},\boldsymbol{S}_d) - \bar{I}(\boldsymbol{U};\boldsymbol{S}_e|\boldsymbol{A}) &= \ubar{I}(\boldsymbol{A},\boldsymbol{U};\boldsymbol{Y},\boldsymbol{S}_d) - \bar{I}(\boldsymbol{A},\boldsymbol{U};\boldsymbol{S}_e|\boldsymbol{A})\notag \\
  \label{eq: equivalent expression general form}&=\ubar{I}(\boldsymbol{U}';\boldsymbol{Y},\boldsymbol{S}_d) - \bar{I}(\boldsymbol{U}';\boldsymbol{S}_e|\boldsymbol{A})
\end{align}
where $\boldsymbol{U}'=(\boldsymbol{U},\boldsymbol{A})$. Let $(\boldsymbol{A},\boldsymbol{U}',\boldsymbol{S}_e,\boldsymbol{S},\boldsymbol{S}_d,\boldsymbol{X},\boldsymbol{Y},\boldsymbol{Z})$ be the random processes as defined in \eqref{eq: equivalent expression general form}. It is obvious that the conditional relations of these random processes remain unchanged and the distortion value is the same since the conditional probability of $X^n$ is $P_{X^n|U^{'n}S^n_e}=P_{X^n|U^nS^n_e}$ in this case because of the Markov chain relation $A^n-(U^n,S^n_e)-X^n$ defined by the joint distributions in Theorems \ref{the: average distortion capacity} and \ref{the: maximal distortion capacity}. Hence, we have $(\boldsymbol{A},\boldsymbol{U}',\boldsymbol{S}_e,\boldsymbol{S},\boldsymbol{S}_d,\boldsymbol{X},\boldsymbol{Y},\boldsymbol{Z})\in\mathcal{P}_D(\text{or}\; \bar{\mathcal{P}}_D)$ for average distortion case (or maximal distortion case). This completes the proof.
\end{remark}
\begin{remark}
  Our results cover existing results for different models. If we remove the distortion constraints in Theorems \ref{the: average distortion capacity} and \ref{the: maximal distortion capacity} and fix the action to some constants, capacity results in both theorems reduce to the capacity of Gel'fand-Pinsker channel with imperfect side information at encoder and decoder\cite[Remark 4]{tan2014formula}. Further setting $\boldsymbol{S}_d=\emptyset$ and $\boldsymbol{S}_e = \boldsymbol{S}$, we recover the general capacity formula of Gel'fand-Pinsker channel\cite[Theorem 1]{tan2014formula},
  \begin{align*}
    C = \sup_{\boldsymbol{U}-(\boldsymbol{X},\boldsymbol{S})-\boldsymbol{Y}}\ubar{I}(\boldsymbol{U};\boldsymbol{Y})-\bar{I}(\boldsymbol{U};\boldsymbol{S}).
  \end{align*}
  For the case that the channel is not state-dependent, by setting $\boldsymbol{A}=\boldsymbol{S}=\boldsymbol{S}_e=\boldsymbol{S}_d=\emptyset$, we recover the general formula of the point-to-point channel\cite{verdu1994general},
  \begin{align*}
    C=\sup_{\boldsymbol{X}} \ubar{I}(\boldsymbol{X};\boldsymbol{Y}).
  \end{align*}
  We further prove in Corollary \ref{coro: discrete memoryless capacity} that our results can also reduce to capacity-distortion tradeoff formula for the discrete memoryless case.
\end{remark}
\begin{remark}
For a given distortion constraint $D$, maximal distortion is a more stringent condition than average distortion. It is proved in \cite{koga2013information} that for lossy source coding with a given distortion constraint, average distortion achieves a lower compression rate. One can always construct a lossy coding code for the average distortion constraint given a code for the maximal distortion constraint. Intuitively, the same result holds for our capacity-distortion problem. That is, given distortion constraint $D$, the capacity under the average constraint should be larger than that under the maximal distortion constraint. However, in general, we cannot construct an average distortion code given a maximal distortion code.

In the lossy compression problem, the encoder has the codebook and observes the source sequence. As shown in \cite[Chapter 5.4]{koga2013information}, the encoder compresses the source sequence with the help of a reference codeword. It uses a codeword in the codebook if it achieves the lowest distortion. Otherwise, the reference codeword is used. This coding scheme does not apply to our problem directly because the state estimator cannot observe the state source sequence, and the reconstruction relies on the channel feedback $Z^n$, which brings uncertainty to the reconstruction process. Even if there exists an appropriate reference codeword, the state estimator does not know if the reconstructed sequence is good enough to output or if it should use the reference codeword.

If we make one more (impractical) assumption that the state estimator is informed of the distortion value, we can construct an average distortion code given a maximal distortion code by modifying the reconstruction function and following the same steps in the proof of \cite[Theorem 5.3.1]{koga2013information}.
\end{remark}
In the next subsection, we prove that our results can be reduced to stationary and memoryless cases.
\subsection{Discrete Memoryless Case}\label{sec: memoryless}
This section considers the case that the states and channels are memoryless. For simplicity, we assume $S^n_d=\emptyset$.
Suppose the distortion function $d_n$ is additive such that
\begin{align}
  \label{def: additive distortion function}\frac{1}{n}d_n(S^n,\hat{S}^n) = \frac{1}{n}\sum_{i=1}^n d(S_i,\hat{S}_i) \leq D_{max},
\end{align}
for some function $d:\mathcal{S}\times\hat{\mathcal{S}}\to [0,+\infty)$.
Suppose the channel states are generated in a stationary and discrete memoryless way and the channel is a stationary and discrete memoryless channel, i.e.,
\begin{align}
  \label{eq: discrete state condition}P_{S^n_eS^nS^n_d|A^n}(s^n_e,s^n,s^n_d|a^n) = \prod_{i=1}^{n} P_{S_{e}SS_{d}|A}(s_{e,i},s_i,s_{d,i}|a_i),\\
  \label{eq: discrete channel condition}P_{Y^nZ^n|X^nS^n}(y^n,z^n|x^n,s^n) = \prod_{i=1}^{n} P_{YZ|XS}(y_i,z_i|x_i,s_i).
\end{align}
In this case, similar to \cite[Lemma 1]{ahmadipour2022information}, the best estimator is 
\begin{align*}
  g(a^n,x^n,s^n_e,z^n) = (g^*(a_1,x_1,s_{e,1},z_1),g^*(a_2,x_2,s_{e,2},z_2),\dots,g^*(a_n,x_n,s_{e,n},z_n)),
\end{align*}
where
\begin{align}
  \label{def: memoryless best state estimator}g^*(a,x,s_e,z) := \mathop{\arg\min}_{\hat{s}} \sum_{s\in\mathcal{S}}P_{S|AXS_eZ}(s|a,x,s_e,z)d(s,\hat{s})
\end{align}
and
\begin{align*}
  P_{S|AXS_eZ}(s|a,x,s_e,z) &= \frac{P_{ZS_eS|XA}(z,s_e,s|x,a)}{\sum_{s'}P_{ZS_eS|XA}(z,s_e,s'|x,a)}\\
  &=\frac{P_{S_eS|A}(s_e,s|a)P_{Z|XS}(z|x,s)}{\sum_{s'}P_{S_eS|A}(s_e,s'|a)P_{Z|XS}(z|x,s')}
\end{align*}
only depending on $P_{S_eS|A}$ and $P_{Z|XS}.$ Now given $P_{S_eS|A}$ and $P_{YZ|XS}$, define
\begin{align*}
  &\mathcal{P}_{M,D} = \left\{(P_A,P_{U|AS_e}P_{X|US_e})| P_{AS_eUX}(a,s_e,u,x)=P_A(a)P_{S_e|A}(s_e|a)P_{U|AS_e}(u|a,s_e)P_{X|US_e}(x|u,s_e),\right.\\
  &\left. \quad\quad\quad\quad\quad\quad\quad\quad \sum_{a,s_e,u,x}P_{AS_eUX}(a,s_e,u,x)\mathbb{E}[d(S,g^*(A,X,S_e,Z))|a,x,s_e]\leq D\right\},
\end{align*}
where `M' stands for memoryless channel and `D' stands for discrete channel.

For the case that the states and channels are memoryless but nonstationary, the best estimator is a set of symbolwise estimators $\{g^*_i\}_{i=1}^{\infty}$ such that
\begin{align}
  \label{def: memoryless and nonstationary best state estimator}g^*_i(a,x,s_e,z) := \mathop{\arg\min}_{\hat{s}} \sum_{s\in\mathcal{S}}P_{S_i|A_iX_iS_{e,i}Z_i}(s|a,x,s_e,z)d(s,\hat{s})
  \end{align}
  and
  \begin{align*}
    P_{S_i|A_iX_iS_{e,i}Z_i}(s|a,x,s_e,z) &= \frac{P_{Z_iS_{e,i}S_i|X_iA_i}(z,s_e,s|x,a)}{\sum_{s'}P_{Z_iS_{e,i}S_i|X_iA_i}(z,s_e,s'|x,a)}\\
    &=\frac{P_{S_{e,i}S_i|A_i}(s_e,s|a)P_{Z_i|X_iS_i}(z|x,s)}{\sum_{s'}P_{S_{e,i}S_i|A_i}(s_e,s'|a)P_{Z_i|X_iS_i}(z|x,s')}
  \end{align*}
  only depending on $P_{S_{e,i}S_i|A_i}$ and $P_{Z_i|X_iS_i}.$ Further define 
  \begin{align*}
    \mathcal{P}^n_{NM,D} &= \left\{(P_{A^n},P_{U^n|A^nS^n_e}P_{X^n|U^nS^n_e})\bigg| P_{A^n}(a^n)=\prod_{i=1}^{n}P_{A_i}(a_i), \right.\\
    &\quad\quad\quad\quad\quad\quad\quad\quad\quad\quad P_{U^n|A^nS^n_e}(u^n|a^n,s^n_e)P_{X^n|U^nS^n_e}(x^n|u^n,s^n_e)=\prod_{i=1}^{n} P_{U_i|A_iS_{e,i}}(u_i|a_i,s_{e,i})P_{X_i|U_iS_{e,i}}(x_i|u_i,s_{e,i}), \\
    &\quad\quad\quad\quad\quad\quad\quad\quad\quad\quad \left.\frac{1}{n}\sum_{i=1}^n\mathbb{E}\left[ d(S_i,g^*_i(A_i,X_i,S_{e,i},Z_i)) \right]\leq D\right\},
  \end{align*}
  where $`NM'$ in the subscript stands for the nonstationary and memoryless property.
\begin{corollary}\label{coro: discrete memoryless capacity}
  Let $\mathcal{A},\mathcal{X},\mathcal{S}_e,\mathcal{S},\mathcal{Y},\mathcal{Z}$ be finite sets. The state distribution $P_{S^e_eS^nS^n_d|A^n}$ and the channel $P_{Y^nZ^n|X^nS^n}$ are stationary and memoryless.
The capacity of stationary and memoryless action-dependent ISAC model is
  \begin{align}
    \label{eq: stationary and memoryless capacity 1}\bar{C}(D) = C(D) &= \max_{(P_A,P_{U|AS_e}P_{X|US_e})\in\mathcal{P}_{M,D}}I(A,U;Y) - I(U;S_e|A)\\
    \label{eq: stationary and memoryless capacity 2}&\overset{(a)}{=} \max_{(P_A,P_{U|AS_e}P_{X|US_e})\in\mathcal{P}_{M,D}}I(U;Y) - I(U;S_e|A)
  \end{align}
  If the states and the channel are nonstationary and memoryless, and either one of the two limits
  \begin{align*}
    \lim_{n\to\infty}\frac{1}{n}\sum_{i=1}^n I(U_i,A_i;Y_i),\;\;\;\; \lim_{n\to\infty}\frac{1}{n}\sum_{i=1}^n I(S_{e,i};U_i|A_i)
  \end{align*}
  exists, the capacity of the memoryless but nonstationary action-dependent ISAC model is
  \begin{align*}
    \liminf_{n\to\infty}\max_{(P_{A^n},P_{U^n|A^nS^n_e}P_{X^n|U^nS^n_e})\in\mathcal{P}^n_{NM,D}}\frac{1}{n}\sum_{i=1}^n I(U_i,A_i;Y_i)-\frac{1}{n}\sum_{i=1}^n I(S_{e,i};U_i|A_i)
  \end{align*}
\end{corollary}
The proof is provided in Appendix \ref{app: proof of discrete memoryless capacity}.
\begin{remark}
  Here we prove that similar to \cite{weissman2010capacity}, the equality in $(a)$ still holds even with the distortion constraint. First note that for any distribution $(P_A,P_{UX|AS_e})$, $I(A,U;Y) - I(U;S_e|A) \geq I(U;Y) - I(U;S_e|A)$ always holds. For another direction, let $U^*,A^*$ be random variables achieving the maximum in \eqref{eq: stationary and memoryless capacity 1}. It follows that
  \begin{align*}
    I(A^*,U^*;Y) - I(U^*;S_e|A^*) = I(A^*,U^*;Y) - I(A^*,U^*;S_e|A^*) = I(U';Y) - I(U';S_e|A^*),
  \end{align*}
  where $U' = (A^*,U^*)$. Now the new distribution $(P_A,P_{U'X|AS_e})$ satisfies
  \begin{align*}
    &\sum_{a,s_e,u',x}P_A(a)P_{S_e|A}(s_e|a)P_{U'X|AS_e}(u',x|a,s_e)\mathbb{E}[d(S,g^*(A,X,S_e,Z))|a,x,s_e]\\
    &=\sum_{a,s_e,u,\widetilde{a},x}P_A(a)P_{S_e|A}(s_e|a)P_{U^*\widetilde{A}X|AS_e}((u,\widetilde{a}),x|a,s_e)\mathbb{E}[d(S,g^*(A,X,S_e,Z))|a,x,s_e]\\
    &=\sum_{a,s_e,u,x}P_A(a)P_{S_e|A}(s_e|a)P_{U^*X|AS_e}(u,x|a,s_e)\mathbb{E}[d(S,g^*(A,X,S_e,Z))|a,x,s_e]\leq D.
  \end{align*}
  This completes the proof.
\end{remark}
\begin{remark}
  Note that compared to the capacity results for nonstationary and memoryless channels in \cite[Eq (3.2.33)]{koga2013information}\cite[Corollary 3]{tan2014formula}, our capacity result takes the optimization outside the sum operation. This is because we have an overall distortion constraint $\limsup_{n\to\infty}\frac{1}{n}\sum_{i=1}^n \mathbb{E}\left[ d(S_i,\hat{S}_i) \right]\leq D$, and the overall distortion constraint does not imply each of its components satisfies $\mathbb{E}\left[ d(S_i,\hat{S}_i) \right]\leq D$.
\end{remark}
\subsection{Mixed Channel and States Case}\label{sec: mixed case}
In this section, we consider the action-dependent ISAC model where state and channel are mixed. The state and channel are described by
\begin{align}
  \label{def: mixed state distribution}P_{S^n_eS^nS^n_d|A^n} = \alpha_1 P_{S^n_{e,1}S^n_1S^n_{d,1}|A^n} + \alpha_2P_{S^n_{e,2}S^n_2S^n_{d,2}|A^n},\\
  \label{def: mixed channel distribution}P_{Y^nZ^n|X^nS^n} = \beta_1 P_{Y^n_1Z^n_1|X^nS^n} + \beta_2 P_{Y^n_2Z^n_2|X^nS^n},
\end{align}
where $\alpha_1 + \alpha_2=1,\beta_1+\beta_2=1$. Note here the subscripts $1$ and $2$ represent two different conditional distributions generating the states. For the states, they are still generated by the same action sequence and for the channels they have the same input sequences.

\subsubsection{Mixed Case for Maximal Distortion}

\begin{theorem}\label{the: mixed maximal distortion}
  The capacity of the mixed action-dependent ISAC model is
\begin{align}
  \label{def: capacity of mixed maximal distortion}C_{m}(D) = \sup\min_{i\in\{1,2\},j\in\{1,2\}}\ubar{I}(\boldsymbol{A},\boldsymbol{U}_i;\boldsymbol{S}_{d,i},\boldsymbol{Y}_{i,j}) - \max_{i\in\{1,2\}}\bar{I}(\boldsymbol{U}_i;\boldsymbol{S}_{e,i}|\boldsymbol{A}).
\end{align}
where the sup is taken over random processes $(\boldsymbol{A},\boldsymbol{U},\boldsymbol{X},\boldsymbol{S}_{e},\boldsymbol{S},\boldsymbol{S}_{d},\boldsymbol{Y},\boldsymbol{Z})=\{(A^n,U^n,X^n,S^n_e,S^n,S^n_d,Y^n,Z^n,\hat{S}^n)\}_{n=1}^{\infty}$ with state distribution satisfying \eqref{def: mixed state distribution} and channel distribution satisfying \eqref{def: mixed channel distribution} and $\bar{D}(\boldsymbol{S},\hat{\boldsymbol{S}})\leq D$. 

If the channel states and the channel are stationary and memoryless and the distortion function $d_n$ is additive, we have the single-letter achievable rate
\begin{align}
  \label{def: rate of mixed memoryless}C_{m}(D) \geq \max_{\substack{P_A,P_{U|AS_e}P_{X|US_e}:\\ \mathbb{E}\left[ d(S,\hat{S}) \right] \leq D}} \min_{i\in\{1,2\},j\in\{1,2\}}I(A,U_i;S_{d,i},Y_{i,j}) - \max_{i\in\{1,2\}}I(U_i;S_{e,i}|A).
\end{align}  
where the maximization is over all distribution $P_{AS_eSS_dUXYZ\hat{S}}$ with $P_{S_eSS_d|A}=\alpha_1 P_{S_{e,1}S_1S_{d,1}|A} + \alpha_2P_{S_{e,2}S_2S_{d,2}|A}$ and channel distribution $P_{YZ|XS}=\beta_1 P_{Y_1Z_1|XS} + \beta_2 P_{Y_2Z_2|XS}$ and distortion $\mathbb{E}\left[ d(S,\hat{S}) \right] \leq D$.
\end{theorem}
The proof is given in Appendix \ref{app: proof of mixed case results}.

\subsubsection{Mixed Case for Average Distortion Case}

For the average distortion case, it is more convenient to write the result in distortion-capacity tradeoff form.
\begin{theorem}\label{the: mixed average distortion}
  The distortion-capacity formula in this case is bounded by
\begin{align}
  \label{def: capacity of mixed average distortion} D_{a}(C) \leq \sup \sum_{i=1}^2\sum_{j=1}^2 \alpha_i\beta_j D(\boldsymbol{S}_i,\hat{\boldsymbol{S}}_{i,j}),
\end{align}
where the sup is taken over random processes $(\boldsymbol{A},\boldsymbol{U},\boldsymbol{X},\boldsymbol{S}_{e},\boldsymbol{S},\boldsymbol{S}_{d},\boldsymbol{Y},\boldsymbol{Z})=\{(A^n,U^n,X^n,S^n_e,S^n,S^n_d,Y^n,Z^n,\hat{S}^n)\}_{n=1}^{\infty}$ with state distribution satisfying \eqref{def: mixed state distribution} and channel distribution satisfying \eqref{def: mixed channel distribution} and $$\min_{i\in\{1,2\},j\in\{1,2\}}\ubar{I}(\boldsymbol{A},\boldsymbol{U}_i;\boldsymbol{S}_{d,i},\boldsymbol{Y}_{i,j}) - \max_{i\in\{1,2\}}\bar{I}(\boldsymbol{U}_i;\boldsymbol{S}_{e,i}|\boldsymbol{A})\leq C.$$
\end{theorem}

Similar to the lossy compression case\cite[Theorem 5.10.2]{koga2013information}, the result in \eqref{def: capacity of mixed average distortion} is an upper bound. In \cite[Ch. 5]{koga2013information}, it was proved that for lossy source coding case, the equality holds when at least one of the sources is stationary and the distortion function $d_n$ is subadditive. The technique does not fit the model we consider here since the optimization in \eqref{def: capacity of mixed average distortion} is taken over all the terms jointly. Hence, the condition for the equality holding in \eqref{def: capacity of mixed average distortion} is still an open problem.
\subsection{Rate-limited CSI}\label{sec: rate-limited csi}
In this section, we first consider the case that the imperfect side information $\boldsymbol{S}_d$ is available at the decoder side and the rate-limited version of $\boldsymbol{S}_d$ available at the encoder side. Then, the reversed case with imperfect CSI at the encoder side and rate-limited CSI at the decoder side is considered.

\begin{theorem}\label{the: rate-limited CSI at encoder}
  The capacity of the general state-dependent ISAC model with rated limited CSI at encoder under average distortion constraint is the set of pairs $(R_e,R)$ satisfying
  \begin{align*}
    R_e &\geq \bar{I}(\boldsymbol{V};\boldsymbol{S}_d),\\
    R &\leq \ubar{I}(\boldsymbol{A},\boldsymbol{X};\boldsymbol{Y},\boldsymbol{S}_d|\boldsymbol{V}),
  \end{align*}
  with the underlying random processes $(\boldsymbol{A},\boldsymbol{X},\boldsymbol{V},\boldsymbol{S}_d,\boldsymbol{S},\boldsymbol{Y},\boldsymbol{Z})=\{A^n,X^n,V^n,S^n,S^n_d,Y^n,Z^n\}_{n=1}^{\infty}$ and joint distribution $P_{A^n}P_{S^nS^n_d|A^n}$ $P_{V^n|S^n_d}P_{X^n|A^nV^n}P_{Y^nZ^n|X^nS^n}$ and $\limsup_{n\to\infty}\frac{1}{n}\mathbb{E}\left[ d_n(S^n,g_n(X^n,A^n,V^n,Z^n)) \right] \leq D$.
  The capacity of the general state-dependent ISAC model with rate-limited CSI at encoder under maximal distortion constraint is a set of pairs $(R_e,R)$ satisfying
  \begin{align*}
    R_e &\geq \bar{I}(\boldsymbol{V};\boldsymbol{S}_d),\\
    R&\leq\ubar{I}(\boldsymbol{A},\boldsymbol{X};\boldsymbol{Y},\boldsymbol{S}_d|\boldsymbol{V}),
  \end{align*}
  with the underlying  random processes $(\boldsymbol{A},\boldsymbol{X},\boldsymbol{V},\boldsymbol{S}_d,\boldsymbol{S},\boldsymbol{Y},\boldsymbol{Z})=\{A^n,X^n,V^n,S^n,S^n_d,Y^n,Z^n\}_{n=1}^{\infty}$ and joint distribution $P_{A^n}P_{S^nS^n_d|A^n}$ $P_{V^n|S^n_d}P_{X^n|A^nV^n}P_{Y^nZ^n|X^nS^n}$ and $p-\limsup_{n\to\infty}\frac{1}{n} d_n(S^n,g_n(X^n,A^n,V^n,Z^n)) \leq D$.
\end{theorem}

The coding scheme is similar to that in Sections \ref{sec: proof of average distortion} and \ref{sec: proof of maximal distortion}. However, the random binning technique is no longer needed since the encoder only observes a lossy version of the imperfect side information $\boldsymbol{S}_d$. Since the decoder has full information of $\boldsymbol{S}_d$, this lossy description at the encoder can be regarded as the common information at both the encoder and decoder. The proof is provided in Appendix \ref{app: proof of rate-limited csi at encoder}.

\begin{theorem}\label{the: rate-limited CSI at decoder}
  The capacity of the general state-dependent ISAC model with rated limited CSI at decoder under average distortion constraint is the set of pairs $(R_d,R)$ satisfying
  \begin{align*}
    R_d &\geq \bar{I}(\boldsymbol{V};\boldsymbol{S}_e) - I(\boldsymbol{V};\boldsymbol{Y}),\\
    R &\leq \ubar{I}(\boldsymbol{A},\boldsymbol{U};\boldsymbol{Y}|\boldsymbol{V}) - \bar{I}(\boldsymbol{U};\boldsymbol{S}_e|\boldsymbol{A},\boldsymbol{V}) =  \ubar{I}(\boldsymbol{U};\boldsymbol{Y}|\boldsymbol{V})- \bar{I}(\boldsymbol{U};\boldsymbol{S}_e|\boldsymbol{A},\boldsymbol{V}),
  \end{align*}
  with the underlying random processes $(\boldsymbol{A},\boldsymbol{U},\boldsymbol{V},\boldsymbol{S}_d,\boldsymbol{S},\boldsymbol{X},\boldsymbol{Y},\boldsymbol{Z})=\{A^n,U^n,V^n,S^n,S^n_e,X^n,Y^n,Z^n\}_{n=1}^{\infty}$ and joint distribution $P_{A^n}P_{S^nS^n_e|A^n}P_{V^n|S^n_e}P_{U^n|A^nV^n}P_{X^n|U^nV^n}P_{Y^nZ^n|X^nS^n}$ and $\limsup_{n\to\infty}\frac{1}{n}\mathbb{E}\left[ d_n(S^n,g_n(X^n,A^n,V^n,Z^n)) \right] \leq D$.
  The capacity of the general state-dependent ISAC model with rate-limited CSI at encoder under maximal distortion constraint is the set of pairs $(R_d,R)$ satisfying
  \begin{align*}
    R_d &\geq \bar{I}(\boldsymbol{V};\boldsymbol{S}_e) - I(\boldsymbol{V};\boldsymbol{Y}),\\
    R&\leq\ubar{I}(\boldsymbol{A},\boldsymbol{U};\boldsymbol{Y}|\boldsymbol{V}) - \bar{I}(\boldsymbol{U};\boldsymbol{S}_e|\boldsymbol{A},\boldsymbol{V}) = \ubar{I}(\boldsymbol{U};\boldsymbol{Y}|\boldsymbol{V})- \bar{I}(\boldsymbol{U};\boldsymbol{S}_e|\boldsymbol{A},\boldsymbol{V}),
  \end{align*}
  with the underlying  random processes $(\boldsymbol{A},\boldsymbol{U},\boldsymbol{V},\boldsymbol{S}_e,\boldsymbol{S},\boldsymbol{X},\boldsymbol{Y},\boldsymbol{Z})=\{A^n,U^n,V^n,S^n,S^n_e,X^n,Y^n,Z^n\}_{n=1}^{\infty}$ and joint distribution $P_{A^n}P_{S^nS^n_e|A^n}P_{V^n|S^n_e}P_{U^n|A^nV^n}P_{X^n|U^nV^n}P_{Y^nZ^n|X^nS^n}$ and $p-\limsup_{n\to\infty}\frac{1}{n} d_n(S^n,g_n(X^n,A^n,V^n,Z^n)) \leq D$.
\end{theorem}
\begin{remark}
  By setting $\boldsymbol{S}_e=\boldsymbol{S},\boldsymbol{A}=\emptyset$ and removing the distortion constraint, we recover the capacity result provided in \cite[Theorem 6]{tan2014formula} for the general Gel'fand-Pinsker channel with rate-limited CSI at the decoder.
\end{remark}
For the direct part, applying Wyner-Ziv coding for general sources \cite{iwata2002information} for source $S^n_e$ with $Y^n$ being the correlated side information gives the bound of $R_d$. The lossy description $V^n$ is sent to the decoder. Now similar to Theorem \ref{the: rate-limited CSI at encoder}, the lossy description $V^n$ is regarded as the common information at both the encoder and decoder. The remaining proofs are the conditional version of the proofs in Sections \ref{sec: proof of average distortion} and \ref{sec: proof of maximal distortion} by removing $S^n_d$ at the decoder side and including common information $V^n$.

For the converse part, we again let $U^n$ represent the uniformly distributed message and let $A^n=f_A(M)=f_A(U^n)$. The input sequence $X^n=f(M,V^n)=f(U^n,V^n)$, where $V^n$ is recoverable at both the encoder side and the decoder side. It is obvious that $U^n-(A^n,V^n)-S^n_e$ forms a Markov chain for $n=1,2,\dots$ and hence $\bar{I}(\boldsymbol{U};\boldsymbol{S}_e|\boldsymbol{A},\boldsymbol{V})=0$. Now following the same argument as in \cite[Proof of Theorem 6]{tan2014formula} completes the converse part.

\section{proof of theorem \ref{the: average distortion capacity}}\label{sec: proof of average distortion}
In this section, we provide the forward and converse part of the proof of Theorem \ref{the: average distortion capacity}. The reliable part analysis is similar to the analysis in \cite{tan2014formula}. However, we need to choose the codewords more carefully because of the additional distortion constraint. Before giving the coding scheme, we first give some auxiliary functions that play an important role in the selection of the codewords. Then we give the coding scheme and the corresponding reliable and distortion analysis. 

Consider input random variables $(A^n,U^n,X^n,$ $S^n_e,S^n,S^n_d,Y^n)$ with joint distribution $P_{A^n}P_{S^n_eS^nS^n_d|A^n}P_{U^n|A^nS^n_e}P_{X^n|U^nS^n_e}P_{Y^n|X^nS^n}$ such that the average distortion satisfies $\frac{1}{n}\mathbb{E}\left[ d_n(S^n,g(X^n,A^n,S^n_e,Z^n)) \right] \leq D$.

For a given codebook $\mathcal{C}$, define the codeword selection mapping $F_n^{\mathcal{C}}:\mathcal{A}^n\times\mathcal{S}^n_e \to \mathcal{U}^n$ and auxiliary functions
 $\eta_1: \mathcal{U}^n \times \mathcal{A}^n \times \mathcal{S}_e^n \to \mathbb{R}^{+}$ and $\eta_2: \mathcal{U}^n \times \mathcal{A}^n \times \mathcal{S}_e^n \to \mathbb{R}^{+}$ as
\begin{align*}
  \eta_1(u^n,a^n,s^n_e) := \sum_{x^n,s^n,z^n} \sum_{\substack{(y^n,s^n_d):\\(u^n,a^n,s^n_d,y^n)\notin\mathcal{T}_1}} P_{S^nS^n_d|A^nS^n_e}(s^n,s^n_d|a^n,s^n_e)P_{X^n|U^nS^n_e}(x^n|u^n,s^n_e)P_{Y^nZ^n|X^nS^n}(y^n,z^n|x^n,s^n),\\
  \eta_2(u^n,a^n,s^n_e) := \sum_{s^n,s^n_d}P_{S^nS^n_d|A^nS^n_e}(s^n,s^n_d|a^n,s^n_e)\sum_{x^n}P_{X^n|U^nS^n_e}(x^n|u^n,s^n_e) \sum_{y^n,z^n}P_{Y^nZ^n|X^nS^n}(y^n,z^n|x^n,s^n)d_n(s^n,g(a^n,x^n,s^n_e,z^n)),
\end{align*}
where
\begin{align*}
  \mathcal{T}_1 = \left\{(u^n,a^n,s^n_d,y^n):\frac{1}{n}\log \frac{P_{Y^nS^n_d|A^nU^n}(y^n,s^n_d|a^n,u^n)}{P_{Y^nS^n_d}(y^n,s^n_d)} \geq \ubar{I}(\boldsymbol{A},\boldsymbol{U};\boldsymbol{Y},\boldsymbol{S}_d) - \gamma\right\}
\end{align*}
for some arbitrary but fixed $\gamma>0$. The function $\eta_1$ gives the probability that the channel output does not fall into the set $\mathcal{T}_1$, which is defined as the decoding set in our coding scheme. Function $\eta_2(u^n,a^n,s^n_e)$ is the expectation value of the distortion given $(u^n,a^n,s^n_e)$.
Further, define
\begin{align*}
  \mathcal{T}_2 = \left\{ (u^n,a^n,s^n_e): \frac{1}{n}\log\frac{P_{U^n|A^nS^n_e}(u^n|a^n,s^n_e)}{P_{U^n|A^n}(u^n|a^n)} \leq \bar{I}(\boldsymbol{U};\boldsymbol{S}_e|\boldsymbol{A}) + \gamma \right\}.
\end{align*}
Set 
\begin{align*}
  \pi_1 &:= Pr\{(U^n,A^n,S^n_d,Y^n)\notin \mathcal{T}_1\}\to 0 \;\;\text{as $n\to\infty$},\\
  \pi_2 &:= Pr\{(U^n,A^n,S^n_e)\notin \mathcal{T}_2\}\to 0\;\; \text{as $n\to\infty$},\\
  \mathcal{B} &:= \{(u^n,a^n,s^n_e): \eta_1(u^n,a^n,s^n_e) \leq \pi_1^{\frac{1}{2}}\}.
\end{align*}
By Markov's inequality \cite[Theorem 1.6.4]{durrett2019probability}, we have
\begin{align*}
  Pr\{(U^n,A^n,S^n_e)\notin \mathcal{B}\} \leq \pi_1^{\frac{1}{2}}.
\end{align*}

\emph{Codebook Generation.} Generate action codebook $\mathcal{A}=\{a^n(m):m\in[1:2^{nR}]\}.$ For each message $m\in[1:2^{nR}]$, generate a subcodebook $\mathcal{C}(m)=\{u^n(m,l):l\in[(m-1)2^{nR'}+1:m2^{nR'}]\}$ with $R'=\bar{I}(\boldsymbol{U};\boldsymbol{S}_e|\boldsymbol{A}) + 2\gamma, \gamma>0$, each according to distribution $P_{U^n|A^n}(\cdot | a^n(m))$. The whole codebook is defined with $\mathcal{C}=\cup_{m\in\mathcal{M}}\mathcal{C}(m).$ 

\emph{Encoding. } To transmit message $m$, the encoder selects action sequence $a^n(m)$, and then observes the imperfect state sequence $s^n_e$. The encoder looks for a sequence $u^n(m,l)\in\mathcal{C}(m)$ such that
\begin{align*}
  (u^n(m,l),a^n(m),s^n_e)\in \mathcal{B}.
\end{align*}
 If there are more than one $u^n(m,l)$, choose the index $l^*$ such that
\begin{align*}
  u^n(m,l^*)=\mathop{\arg\min}_{u^n(m,l)\in\mathcal{C}(m):(u^n(m,l),a^n,s^n_e)\in\mathcal{B}}\eta_2(u^n,a^n,s^n_e).
\end{align*}
If no such $l$ exists, choose
\begin{align*}
  u^n(m,l^*)=\mathop{\arg\min}_{u^n(m,l)\in\mathcal{C}(m)}\eta_2(u^n,a^n,s^n_e).
\end{align*}
Given a random codebook $\mathbf{C}$, define this codeword selection mapping by $F_n^{\mathbf{C}}:\mathcal{A}^n \times \mathcal{S}^n_e \to \mathbf{C}$. The randomness of the mapping comes from the random codebook $\mathbf{C}$ and the imperfect state sequence $S^n_e$.

\emph{Decoding.} Given $y^n\in\mathcal{Y}^n, s^n_d\in\mathcal{S}^n_d$, the decoder looks for a unique message $\hat{m}$ such that there exist some $u^n(\hat{m},l)\in\mathcal{C}(\hat{m})$ satisfying
\begin{align*}
  (a^n(\hat{m}),u^n(\hat{m},l),s^n_d,y^n)\in\mathcal{T}_1.
\end{align*}
If there is no such unique message $\hat{m}$, the decoder declares an error.

\textbf{Error analysis: }
Due to symmetry, it is sufficient to consider the case that $m=1$ is sent and $L$ is the index selected by the encoder.
Define events
\begin{align*}
  \mathcal{E}_1 = \{(A^n(1),U^n(1,l),S^n_e)\notin \mathcal{B} \;\text{for all $U^n(1,l)\in\mathcal{C}(1)$}\},\\
  \mathcal{E}_2 = \{(A^n(1),U^n(1,L),S^n_d,Y^n)\notin \mathcal{T}_1  \},\\
  \mathcal{E}_3 = \{\exists \widetilde{m}\neq 1: (A^n(\widetilde{m}),U^n(\widetilde{m},\widetilde{l}),S^n_d,Y^n)\in \mathcal{T}_1 \}.
\end{align*}
It follows that $P_{e,n} \leq Pr\{\mathcal{E}_1\} + Pr\{\mathcal{E}_2\cap \mathcal{E}_1^c\} + Pr\{\mathcal{E}_3\}$. By setting $R' = \bar{I}(\boldsymbol{U};\boldsymbol{S}_e|\boldsymbol{A}) + 2\gamma$, we have
\begin{align*}
  Pr\{\mathcal{E}_1\} \leq Pr\{(A^n,U^n,S^n_e)\notin \mathcal{B}\} + Pr\{(A^n,U^n,S^n_e)\notin \mathcal{T}_2\} + \exp(-\exp(n\gamma)) \leq \pi_1^{\frac{1}{2}} + \pi_2 + \exp(-\exp(n\gamma)).
\end{align*}
The proof is the same as the bounding of \cite[Eq. (58)]{tan2014formula} and hence is omitted. For $\mathcal{E}_2$, we define event $\mathcal{F}$, which is analogous to \cite[Eq. (77)]{tan2014formula}:
\begin{align*}
  \mathcal{F} := \{(F^{\textbf{C}}_n(A^n,S^n_e),A^n,S^n_e)\in\mathcal{B}\}
\end{align*}
and it follows that 
$\Pr\{\mathcal{E}_2\cap \mathcal{E}_1^c\}=\Pr\{\mathcal{E}_2\cap \mathcal{E}_1^c \cap \mathcal{F}\} + \Pr\{\mathcal{E}_2\cap\mathcal{E}_1^c\cap \mathcal{F}^c\}\leq \Pr\{\mathcal{E}_2\cap \mathcal{F}\} + \Pr\{\mathcal{E}_1^c\cap \mathcal{F}^c\}$. The probability $\Pr\{\mathcal{E}_1^c\cap \mathcal{F}^c\}$ is zero by the definition of $\mathcal{E}_1$.

Now we have the following bound.
\begin{align*}
  &\Pr\{\mathcal{E}_2\cap \mathcal{E}_1^c\} \\
  &\leq \Pr\{\mathcal{E}_2\cap \mathcal{F}\}\\
  &\leq Pr\left\{ \{(F^{\textbf{C}}_n(A^n,S^n_e),A^n,S^n_e)\in\mathcal{B}\} \cap \{(F^{\textbf{C}}_n(A^n,S^n_e),A^n,Y^n,S^n_d)\in\mathcal{T}_1\} \right\}\\
  &=\sum_{a^n}P_{A^n}(a^n)\mathbb{E}_{\mathbf{C}}\Bigg[ \sum_{\substack{(s^n_e,s^n_d,y^n):\\(F^{\textbf{C}}_n(a^n,s^n_e),a^n,s^n_e)\in \mathcal{B},\\(F^{\textbf{C}}_n(a^n,s^n_e),a^n,y^n,s^n_d)\notin \mathcal{T}_1}} \sum_{s^n,x^n,z^n}P_{S^n_eS^nS^n_d|A^n}(s^n_e,s^n,s^n_d|a^n)\Big.\\
  &\quad\quad\quad\quad\quad\quad\quad\quad \Big. P_{X^n|U^nS^n_e}(x^n|s^n_e,F^{\textbf{C}}_n(a^n,s^n_e))P_{Y^nZ^n|X^nS^n}(y^n,z^n|x^n,s^n) \Bigg]\\
  &=\sum_{a^n}P_{A^n}(a^n)\mathbb{E}_{\mathbf{C}}\Bigg[ \sum_{\substack{s^n_e:\\(F^{\textbf{C}}_n(a^n,s^n_e),a^n,s^n_e)\in \mathcal{B}}} P_{S^n_e|A^n}(s^n_e|a^n)\Bigg.\\
  &\quad\quad\quad\quad \Bigg. \sum_{\substack{(s^n_d,y^n):\\(F^{\textbf{C}}_n(a^n,s^n_e),a^n,y^n,s^n_d)\notin \mathcal{T}_1}}\sum_{s^n,x^n,z^n}P_{S^nS^n_d|A^nS^n_e}(s^n,s^n_d|a^n,s^n_e)P_{X^n|U^nS^n_e}(x^n|s^n_e,F^{\textbf{C}}_n(a^n,s^n_e))P_{Y^nZ^n|X^nS^n}(y^n,z^n|x^n,s^n) \Bigg]\\
  &\leq \sum_{a^n}P_{A^n}(a^n)\mathbb{E}_{\mathbf{C}}\Bigg[ \sum_{\substack{s^n_e:\\(F^{\textbf{C}}_n(a^n,s^n_e),a^n,s^n_e)\in \mathcal{B}}} P_{S^n_e|A^n}(s^n_e|a^n) \eta_1(F^{\textbf{C}}(a^n,s^n_e),a^n,s^n_e) \Bigg]\overset{(a)}{\leq} \pi_1^{\frac{1}{2}},
\end{align*}
where $(a)$ follows by the definition of $\mathcal{B}$.
The bound of $\mathcal{E}_3$ is similar to that in \cite{tan2014formula} and we have $Pr\{\mathcal{E}_3\}\to 0$ for $n\to\infty$ given $\widetilde{R}= \ubar{I}(\boldsymbol{A},\boldsymbol{U};\boldsymbol{Y},\boldsymbol{S}_d)-2\gamma.$

For the distortion, let $\mathcal{A}$ be a sample action codebook and $P_{\textbf{A}}$ be the distribution of the random codebook $\textbf{A}$. It follows that
\begin{align*}
  &\mathbb{E}\left[ d_n(S^n,g(f_A(M),f(M,S^n_e),S^n_e,Z^n))\right] \\
  &=\frac{1}{|\mathcal{M}|}\sum_{m}\sum_{\mathcal{A}}P_{\textbf{A}}(\mathcal{A})\sum_{s^n_e}P_{S^n_e|A^n}(s^n_e|f_A(m))\sum_{\mathcal{C}}P_{\mathbf{C}|A^n}(\mathcal{C}|f_A(m))\sum_{s^n,s^n_d}P_{S^nS^n_d|A^nS^n_e}(s^n,s^n_d|f_A(m),s^n_e) \\
  &\quad\quad\quad\quad\quad\quad\quad\sum_{x^n}P_{X^n|U^nS^n_e}(x^n|s^n_e,F^{\mathcal{C}}_n(f_A(m),s^n_e))\sum_{y^n,z^n}P_{Y^nZ^n|X^nS^n}(y^n,z^n|x^n,s^n)d_n(s^n,g(f_A(m),x^n,s^n_e,z^n))\\
  &=\frac{1}{|\mathcal{M}|}\sum_{m}\sum_{\mathcal{A}}P_{\textbf{A}}(\mathcal{A})\sum_{s^n_e}P_{S^n_e|A^n}(s^n_e|f_A(m))\sum_{\mathcal{C}}P_{\mathbf{C}|A^n}(\mathcal{C}|f_A(m))\eta_2(f_A(m),F^{\mathcal{C}}_n(f_A(m),s^n_e),s^n_e)\\
  &=\frac{1}{|\mathcal{M}|}\sum_{m}\sum_{\mathcal{A}}P_{\textbf{A}}(\mathcal{A})\sum_{s^n_e}P_{S^n_e|A^n}(s^n_e|f_A(m))\sum_{\mathcal{C}}P_{\mathbf{C}|A^n}(\mathcal{C}|f_A(m))\int_{0}^{\infty}\mathbb{I}\left\{ \eta_2(f_A(m),F^{\mathcal{C}}_n(f_A(m),s^n_e),s^n_e) > \beta \right\} d\beta\\
  &=\frac{1}{|\mathcal{M}|}\sum_{m}\sum_{\mathcal{A}}P_{\textbf{A}}(\mathcal{A})\sum_{s^n_e}P_{S^n_e|A^n}(s^n_e|f_A(m))\int_{0}^{\infty}\sum_{\mathcal{C}}P_{\mathbf{C}|A^n}(\mathcal{C}|f_A(m))\mathbb{I}\left\{ \eta_2(f_A(m),F^{\mathcal{C}}_n(f_A(m),s^n_e),s^n_e) > \beta \right\} d\beta
\end{align*}
where
\begin{align*}
  &\sum_{\mathcal{C}}P_{\mathbf{C}|A^n}(\mathcal{C}|f_A(m))\mathbb{I}\left\{ \eta_2(f_A(m),F^{\mathcal{C}}_n(f_A(m),s^n_e),s^n_e) > \beta \right\}\\
  &\overset{(a)}{\leq} \sum_{\mathcal{C}} \prod_{l=1}^{|\mathcal{C}(m)|}P_{U^n|A^n}(u^n(m,l)|f_A(m))\mathbb{I}\{  \eta_2(f_A(m),u^n(m,l),s^n_e) > \beta \cup (f_A(m),u^n(m,l),s^n_e)\notin \mathcal{B} \}\\
  &=\left\{\sum_{u^n}P_{U^n|A^n}(u^n|f_A(m))\mathbb{I}\{  \eta_2(f_A(m),u^n,s^n_e) > \beta \cup (f_A(m),u^n,s^n_e)\notin \mathcal{B} \}\right\}^{|\mathcal{C}(m)|}\\
  &=\left\{1 - \sum_{u^n}P_{U^n|A^n}(u^n|f_A(m))\mathbb{I}\{  \eta_2(f_A(m),u^n,s^n_e) \leq \beta \cap (f_A(m),u^n,s^n_e)\in \mathcal{B} \}\right\}^{|\mathcal{C}(m)|}\\
  &\leq \left\{1 - \sum_{u^n}P_{U^n|A^n}(u^n|f_A(m))\mathbb{I}\{  \eta_2(f_A(m),u^n,s^n_e) \leq \beta \cap (f_A(m),u^n,s^n_e)\in \mathcal{B} \cap \mathcal{T}_2 \}\right\}^{|\mathcal{C}(m)|}\\
  &\overset{(b)}{\leq} \left\{1 - 2^{-n(\bar{I}(\boldsymbol{U};\boldsymbol{S}_e|\boldsymbol{A})+\gamma)}\sum_{u^n}P_{U^n|A^nS^n_e}(u^n|f_A(m),s^n_e)\mathbb{I}\{  \eta_2(f_A(m),u^n,s^n_e) \leq \beta \cap (f_A(m),u^n,s^n_e)\in \mathcal{B} \cap \mathcal{T}_2 \}\right\}^{|\mathcal{C}(m)|}\\
  &\leq 1 + \exp(-|\mathcal{C}(m)|2^{-n(\bar{I}(\boldsymbol{U};\boldsymbol{S}_e|\boldsymbol{A})+\gamma)})  - \sum_{u^n}P_{U^n|A^nS^n_e}(u^n|f_A(m),s^n_e)\mathbb{I}\{  \eta_2(f_A(m),u^n,s^n_e) \leq \beta \cap (f_A(m),u^n,s^n_e)\in \mathcal{B} \cap \mathcal{T}_2 \}
\end{align*}
where $(a)$ follows by the way we define $F^{\mathcal{C}}_n$ in the encoding phase, $(b)$ follows by the definition of $\mathcal{T}_2$.
Then, we have
\begin{align*}
  &\mathbb{E}\left[ d_n(S^n,g(f_A(M),f(M,S^n_e),S^n_e,Z^n))\right] \\
  &=\frac{1}{|\mathcal{M}|}\sum_{m}\sum_{\mathcal{A}}P_{\textbf{A}}(\mathcal{A})\sum_{s^n_e}P_{S^n_e|A^n}(s^n_e|f_A(m))\\
  &\quad\quad\quad\quad\quad\int_{0}^{\infty} \sum_{\mathcal{C}}P_{\mathbf{C}|A^n}(\mathcal{C}|f_A(m))\mathbb{I}\left\{ \eta_2(f_A(m),F^{\mathcal{C}}_n(f_A(m),s^n_e),s^n_e) > \beta \right\} d\beta\\
  &=\frac{1}{|\mathcal{M}|}\sum_{m}\sum_{a^n}P_{A^n}(a^n)\sum_{s^n_e}P_{S^n_e|A^n}(s^n_e|a^n)\\
  &\quad\quad\quad\quad\quad\int_{0}^{\infty} \sum_{\mathcal{C}}P_{\mathbf{C}|A^n}(\mathcal{C}|a^n)\mathbb{I}\left\{ \eta_2(a^n,F^{\mathcal{C}}_n(a^n,s^n_e),s^n_e) > \beta \right\} d\beta\\
  &\leq \int_{0}^{nD_{max}} Pr\{\eta_2(A^n,X^n,S^n_e) > \beta\}d\beta +  nD_{max}\exp(-2^{n\gamma}) + nD_{max}Pr\{(A^n,U^n,S^n_e)\notin \mathcal{B}\} + nD_{max}Pr\{(A^n,U^n,S^n_e)\notin \mathcal{T}_2\}\\
  &=\mathbb{E}\left[ \eta_2(A^n,U^n,S^n_e) \right] + nD_{max}\exp(-2^{n\gamma}) + nD_{max}\pi_{1}^{\frac{1}{2}} + nD_{max}\pi_2\\
  &=\mathbb{E}\left[ d_n(S^n,g(A^n,X^n,S^n_e,Z^n)) \right]  + nD_{max}\exp(-2^{n\gamma}) + nD_{max}\pi_{1}^{\frac{1}{2}} + nD_{max}\pi_2.
\end{align*}
By the fact that $D_{max}\exp(-2^{n\gamma})\to 0, D_{max}\pi_{1}^{\frac{1}{2}} \to 0$ and $D_{max}\pi_2 \to 0$ as $n\to \infty$,
 it follows that
\begin{align*}
  \limsup_{n\to\infty}\frac{1}{n}\mathbb{E}\left[ d_n(S^n,g(f_A(M),f(M,S^n_e),S^n_e,Z^n))\right] \leq \limsup_{n\to\infty}\frac{1}{n}\mathbb{E}\left[ d_n(S^n,g(A^n,X^n,S^n_e,Z^n)) \right] \leq D,
\end{align*}
which completes the achievability proof.

\emph{Converse.} For the converse, we follow the technique used in \cite{tan2014formula} and also \cite{bloch2013strong}. Consider a sequence of $(n,\mathcal{M}_n,\mathcal{A}_n)$ codes defined in Definition \ref{def: code} satisfying
\begin{align*}
  R=\frac{1}{n}\log |\mathcal{M}_n|,\\
  \limsup_{n\to\infty}\frac{1}{n}\mathbb{E}\left[ d_n(S^n,g_n(f(M,S^n_e),f_A(M),S^n_e,Z^n)) \right] \leq D.
\end{align*} 
 
 Now let $U^n$ be a random variable representing the choice of the message in $[1:2^{nR}]$ and $A^n\in\mathcal{A}$ such that $A^n=f_A(M)=f_A(U^n)$. Further define $X^n=f(M,S^n_e)=f(U^n,S^n_e)$. It follows that 
 \begin{align*}
  \limsup_{n\to\infty}\frac{1}{n}\mathbb{E}\left[ d_n(S^n,g_n(X^n,A^n,S^n_e,Z^n)) \right] \leq D.
 \end{align*}

 By the independence of message and state, and the relation between the action and state, we have the joint distribution $P_{U^nA^n}P_{S^n_eS^nS^n_d|A^n}$ $P_{X^n|U^nS^n_e}P_{Y^nZ^n|X^nS^n}.$ Now analogous to \cite[Appendix A]{tan2014formula}, let $\mathcal{I}_{\boldsymbol{U},\boldsymbol{S}_e}$ be set of random processes $(\boldsymbol{A},\boldsymbol{U},\boldsymbol{S}_e,\boldsymbol{S},\boldsymbol{S}_d,\boldsymbol{X},\boldsymbol{Y},\boldsymbol{Z})$ in which each collection of random variables $(A^n,U^n,S^n_e,S^n,S^n_d,X^n,Y^n,Z^n)$ satisfies $P_{U^nA^n}P_{S^n_eS^nS^n_d|A^n}$ $P_{X^n|U^nS^n_e}P_{Y^nZ^n|X^nS^n}$ and $\mathcal{D}_{\boldsymbol{U},\boldsymbol{S}_e}$ be the set such that the joint distribution satisfies $P_{A^n}P_{S^n_eS^nS^n_d|A^n}P_{U^n|A^nS^n_e}P_{X^n|U^nA^nS^n_e}P_{Y^nZ^n|X^nS^n}$.

Now by Verd\'u and Han's converse theorem \cite[Lemma 3.2.2]{koga2013information}, for any code for a general channel
\begin{align*}
  P_{Y^nS^n_d|U^n}(y^n,s^n_d|u^n)=\sum_{a^n,s^n,s^n_e,z^n,x^n} P_{A^n|U^n}(a^n|u^n) P_{S^n_eS^nS^n_d|A^n}(s^n_e,s^n,s^n_d|a^n)P_{X^n|U^nS^n_e}(x^n|u^n,s^n_e)P_{Y^nZ^n|X^nS^n}(y^n,z^n|x^n,s^n),
\end{align*}
the rate $R$ satisfies $R\leq \ubar{I}(\boldsymbol{U};\boldsymbol{Y},\boldsymbol{S}_d) + 2\gamma$, for some arbitrary but fixed $\gamma>0$. Thus, it follows that
\begin{align*}
  R &\leq \ubar{I}(\boldsymbol{U};\boldsymbol{Y},\boldsymbol{S}_d) + 2\gamma\\
  &= \ubar{I}(\boldsymbol{U};\boldsymbol{Y},\boldsymbol{S}_d) - \bar{I}(\boldsymbol{U};\boldsymbol{S}_e|\boldsymbol{A}) + 2\gamma\\
  &\leq \sup_{\mathcal{I}_{\boldsymbol{U},\boldsymbol{S}_e} \cap \mathcal{P}_D} \ubar{I}(\boldsymbol{U};\boldsymbol{Y},\boldsymbol{S}_d) - \bar{I}(\boldsymbol{U};\boldsymbol{S}_e|\boldsymbol{A}) + 2\gamma\\
  &\leq \sup_{\mathcal{D}_{\boldsymbol{U},\boldsymbol{S}_e} \cap \mathcal{P}_D} \ubar{I}(\boldsymbol{U};\boldsymbol{Y},\boldsymbol{S}_d) - \bar{I}(\boldsymbol{U};\boldsymbol{S}_e|\boldsymbol{A}) + 2\gamma\\
  &=\sup_{\mathcal{P}_D} \ubar{I}(\boldsymbol{U};\boldsymbol{Y},\boldsymbol{S}_d) - \bar{I}(\boldsymbol{U};\boldsymbol{S}_e|\boldsymbol{A}) + 2\gamma.
\end{align*}
According to Remark \ref{rem: equivalence of general capacity expressions}, the proof is completed. \hfill \qedsymbol

\section{proof of theorem \ref{the: maximal distortion capacity}}\label{sec: proof of maximal distortion}
This section provides the coding scheme for Theorem \ref{the: maximal distortion capacity}. The error analysis and the converse proof are the same as in Section \ref{sec: proof of average distortion}. Hence, we only present the proof of the distortion part here.

Fix input random variables $(A^n,U^n,X^n,S^n_e,S^n,S^n_d,Y^n,Z^n)$ with joint distribution $P_{A^n}P_{S^n_eS^nS^n_d|A^n}P_{U^n|A^nS^n_e}P_{X^n|U^nS^n_e}P_{Y^nZ^n|X^nS^n}$ such that $\bar{D}(\boldsymbol{S},\hat{\boldsymbol{S}}):=p-\limsup_{n\to\infty}\frac{1}{n} d_n(S^n,g(X^n,A^n,S^n_e,Z^n))  \leq D$.

Define mappings $\eta_1: \mathcal{U}^n \times \mathcal{A}^n \times \mathcal{S}_e^n \to \mathbb{R}^{+}$ and $\eta_2: \mathcal{U}^n \times \mathcal{A}^n \times \mathcal{S}_e^n \to \mathbb{R}^{+}$ as
\begin{align*}
  \eta_1(u^n,a^n,s^n_e) &:= \sum_{x^n,s^n,z^n} \sum_{\substack{(y^n,s^n_d):\\(u^n,a^n,s^n_d,y^n)\notin\mathcal{T}_1}} P_{S^nS^n_d|A^nS^n_e}(s^n,s^n_d|a^n,s^n_e)P_{X^n|U^nS^n_e}(x^n|u^n,s^n_e)P_{Y^nZ^n|X^nS^n}(y^n,z^n|x^n,s^n),\\
  \eta_2(u^n,a^n,s^n_e)&:= \sum_{s^n_d,y^n}\sum_{\substack{(s^n,x^n,z^n):\\(a^n,s^n_e,s^n,x^n,z^n)\notin\mathcal{T}_2}}P_{S^nS^n_d|A^nS^n_e}(s^n,s^n_d|a^n,s^n_e)P_{X^n|U^nS^n_e}(x^n|u^n,s^n_e)P_{Y^nZ^n|X^nS^n}(y^n,z^n|x^n,s^n),
\end{align*}
where
\begin{align*}
  \mathcal{T}_1 = \left\{(u^n,a^n,s^n_d,y^n):\frac{1}{n}\log \frac{P_{Y^nS^n_d|A^nU^n}(y^n,s^n_d|a^n,u^n)}{P_{Y^nS^n_d}(y^n,s^n_d)} \geq \ubar{I}(\boldsymbol{A},\boldsymbol{U};\boldsymbol{Y},\boldsymbol{S}_d) - \gamma\right\},
\end{align*}
for some arbitrary but fixed $\gamma>0$ and
\begin{align*}
  \mathcal{T}_2 =\left\{ (s^n,a^n,s^n_e,x^n,z^n): \frac{1}{n}d_n(s^n,g(a^n,x^n,s^n_e,z^n)) \leq \bar{D}(\boldsymbol{S},\hat{\boldsymbol{S}}) + \gamma  \right\}.
\end{align*}
Further, define
\begin{align*}
  \mathcal{T}_3 =  \left\{ (u^n,a^n,s^n_e): \frac{1}{n}\log\frac{P_{U^n|A^nS^n_e}(u^n|a^n,s^n_e)}{P_{U^n|A^n}(u^n|a^n)} \leq \bar{I}(\boldsymbol{U};\boldsymbol{S}_e|\boldsymbol{A}) + \gamma \right\}
\end{align*}

By the definitions of $p-\limsup_{n\to\infty}$ and $p-\liminf_{n\to\infty}$ in \cite{koga2013information}, we have
\begin{align*}
  Pr\{(U^n,A^n,S^n_d,Y^n)\notin\mathcal{T}_1\}=\pi_1 \to 0,\\
  Pr\{(S^n,A^n,S^n_e,X^n,Z^n)\notin\mathcal{T}_2\}=\pi_2 \to 0,\\
  Pr\{(U^n,A^n,S^n_e)\notin\mathcal{T}_3\}=\pi_3 \to 0
\end{align*}
as $n\to\infty$.
Set 
\begin{align*}
  \mathcal{B}_1 := \{(u^n,a^n,s^n_e): \eta_1(u^n,a^n,s^n_e) \leq \pi_1^{\frac{1}{2}}\}
\end{align*}
and
\begin{align*}
  \mathcal{B}_2 := \{(u^n,a^n,s^n_e): \eta_2(u^n,a^n,s^n_e) \leq \pi_2^{\frac{1}{2}}\}
\end{align*}
By Markov's inequality \cite[Theorem 1.6.4]{durrett2019probability} we have
\begin{align*}
  Pr\{(A^n,U^n,S^n_e)\notin \mathcal{B}_1\} \leq \pi_1^{\frac{1}{2}},Pr\{(A^n,U^n,S^n_e)\notin \mathcal{B}_2\} \leq \pi_2^{\frac{1}{2}}.
\end{align*}

\emph{Codebook Generation.} Generate action codebook $\mathcal{A}=\{a^n(m):m\in[1:2^{nR}]\}.$ For each message $m\in[1:2^{nR}]$, generate a subcodebook $\mathcal{C}(m)=\{u^n(m,l):l\in[(m-1)2^{nR'}+1:m2^{nR'}]\}$ with $R'=\bar{I}(\boldsymbol{U};\boldsymbol{S}_e|\boldsymbol{A}) + 2\gamma$, each according to distribution $P_{U^n|A^n}(\cdot | a^n(m))$. The whole codebook is defined with $\mathcal{C}=\cup_{m\in\mathcal{M}}\mathcal{C}(m).$ 

\emph{Encoding. } To transmit message $m$, the encoder selects action sequence $a^n(m)$, and then observes the imperfect state sequence $s^n_e$. The encoder looks for the sequence $u^n(m,l)\in\mathcal{C}(m)$ such that
\begin{align*}
  (u^n(m,l),a^n(m),s^n_e)\in \mathcal{B} := \mathcal{B}_1 \cap \mathcal{B}_2.
\end{align*}
If no such $l$ exists, choose $u^n(m,1)$. If there is more than one $u^n(m,l)$, select the lowest $l$.
\begin{comment}
choose the index $l$ such that
\begin{align*}
  u^n(m,l)=\mathop{\arg\min}_{u^n(m,l)\in\mathcal{C}(m):(u^n(m,l),s^n_e)\in\mathcal{B}}\eta_2(u^n,a^n,s^n_e).
\end{align*}
\end{comment}
Given a random codebook $\mathbf{C}$, define this codeword selection mapping by $F_n^{\mathbf{C}}:\mathcal{A}^n \times \mathcal{S}^n_e \to \mathbf{C}$. The randomness of the mapping comes from the random codebook $\mathbf{C}$ and the imperfect state sequence $S^n_e$.

\emph{Decoding.} Given $y^n\in\mathcal{Y}^n, s^n_d\in\mathcal{S}^n_d$, the decoder looks for a unique message $\hat{m}$ such that there exist some $u^n(\hat{m},l)\in\mathcal{C}(\hat{m})$ such that
\begin{align*}
  (a^n(\hat{m}),u^n(\hat{m},l),s^n_d,y^n)\in\mathcal{T}_1.
\end{align*}
If there is no such unique message $\hat{m}$ declare an error.

The decoding error analysis is the same as that in Section \ref{sec: proof of average distortion}. For the distortion, let $\mathbb{I}\{\cdot \}$ be the indicator function. By the definition of $\mathcal{B}_2$ and the encoding condition, we have
\begin{align*}
  &Pr\{\frac{1}{n} d_n(S^n,g_n(f(M,S^n_e),f_{A}(M),S^n_e,Z^n))  > \bar{D}(\boldsymbol{S},\hat{\boldsymbol{S}}) + \gamma | \hat{M} = M\}\\
  &\leq \frac{1}{|\mathcal{M}|}\sum_{m}\sum_{\mathcal{A}}P_{\textbf{A}}(\mathcal{A})\sum_{s^n,s^n_e,x^n,z^n}\sum_{s^n_d,y^n} P_{S^n_eS^nS^n_d|A^n}(s^n_e,s^n,s^n_d|f_A(m))\\
  &\quad\quad\quad\quad\quad\quad\quad\quad\quad\quad\cdot\sum_{\mathcal{C}}P_{\textbf{C}|A^n}(\mathcal{C}|f_A(m))P_{X^n|U^nS^n_e}(x^n|F^{\mathcal{C}}(f_A(m),s^n_e),s^n_e)\\
  &\quad\quad\quad\quad\quad\quad\quad\quad\quad\quad\quad\quad \cdot P_{Y^nZ^n|X^nS^n}(y^n,z^n|x^n,s^n)\mathbb{I}\left\{\frac{1}{n} d_n(s^n,g(x^n,f_A(m),s^n_e,z^n)) > \bar{D}(\boldsymbol{S},\hat{\boldsymbol{S}}) + \gamma\right\}\\
  &=\frac{1}{|\mathcal{M}|}\sum_{m}\sum_{\mathcal{A}}P_{\textbf{A}}(\mathcal{A})\sum_{\substack{\mathcal{C}}}P_{\textbf{C}|A^n}(\mathcal{C}|f_A(m))\sum_{s^n_e}P_{S^n_e|A^n}(s^n_e|f_A(m))\sum_{s^n,x^n,z^n}\sum_{s^n_d,y^n}P_{S^nS^n_d|A^nS^n_e}(s^n,s^n_d|f_A(m),s^n_e)\\
  &\quad\quad\quad\quad\quad\quad\quad \cdot P_{X^n|U^nS^n_e}(x^n|F^{\mathcal{C}}(f_A(m),s^n_e),s^n_e)P_{Y^nZ^n|X^nS^n}(y^n,z^n|x^n,s^n)\mathbb{I}\left\{\frac{1}{n} d_n(s^n,g(x^n,f_A(m),s^n_e,z^n)) > \bar{D}(\boldsymbol{S},\hat{\boldsymbol{S}}) + \gamma\right\}\\
  &=\frac{1}{|\mathcal{M}|}\sum_{m}\sum_{\mathcal{A}}P_{\textbf{A}}(\mathcal{A})\sum_{s^n_e}P_{S^n_e|A^n}(s^n_e|f_A(m))\sum_{\mathcal{C}}P_{\textbf{C}|A^n}(\mathcal{C}|f_A(m))\eta_2(f_A(m),F^{\mathcal{C}}(f_A(m),s^n_e),s^n_e)\\
  &=\frac{1}{|\mathcal{M}|}\sum_{m}\sum_{\mathcal{A}}P_{\textbf{A}}(\mathcal{A})\sum_{s^n_e}P_{S^n_e|A^n}(s^n_e|f_A(m))\left(\sum_{\substack{\mathcal{C}:\\(f_A(m),F^{\mathcal{C}}(f_A(m),s^n_e),s^n_e)\in\mathcal{B}_2}}P_{\textbf{C}|A^n}(\mathcal{C}|f_A(m))\eta_2(f_A(m),F^{\mathcal{C}}(f_A(m),s^n_e),s^n_e) \right.\\
  &\quad\quad\quad\quad\quad\quad\quad\quad + \left. \sum_{\substack{\mathcal{C}:\\(f_A(m),F^{\mathcal{C}}(f_A(m),s^n_e),s^n_e)\notin\mathcal{B}_2}}P_{\textbf{C}|A^n}(\mathcal{C}|f_A(m))\eta_2(f_A(m),F^{\mathcal{C}}(f_A(m),s^n_e),s^n_e) \right)\\
  &\overset{(a)}{\leq} \frac{1}{|\mathcal{M}|}\sum_{m}\sum_{\mathcal{A}}P_{\textbf{A}}(\mathcal{A})\sum_{s^n_e}P_{S^n_e|A^n}(s^n_e|f_A(m))\left(\pi_2^{\frac{1}{2}} + \sum_{\substack{\mathcal{C}:\\(f_A(m),F^{\mathcal{C}}(f_A(m),s^n_e),s^n_e)\notin\mathcal{B}_2}}P_{\textbf{C}|A^n}(\mathcal{C}|f_A(m))\eta_2(f_A(m),F^{\mathcal{C}}(f_A(m),s^n_e),s^n_e) \right)\\
  &= \frac{1}{|\mathcal{M}|}\sum_{m}P_{A^n}(a^n)\sum_{s^n_e}P_{S^n_e|A^n}(s^n_e|a^n)\left(\pi_2^{\frac{1}{2}} + \sum_{\substack{\mathcal{C}:\\(a^n,F^{\mathcal{C}}(a^n,s^n_e),s^n_e)\notin\mathcal{B}_2}}P_{\textbf{C}|A^n}(\mathcal{C}|a^n)\eta_2(a^n,F^{\mathcal{C}}(a^n,s^n_e),s^n_e) \right)
\end{align*}
where $(a)$ follows by the definition of $\mathcal{B}_2$. We further have
\begin{align*}
  &\sum_{\substack{\mathcal{C}:\\(a^n,F^{\mathcal{C}}(a^n,s^n_e),s^n_e)\notin\mathcal{B}_2}}P_{\textbf{C}|A^n}(\mathcal{C}|a^n)\eta_2(a^n,F^{\mathcal{C}}(a^n,s^n_e),s^n_e)\\
  &\leq \sum_{\substack{\mathcal{C}:\\(a^n,F^{\mathcal{C}}(a^n,s^n_e),s^n_e)\notin\mathcal{B}_2}}P_{\textbf{C}|A^n}(\mathcal{C}|a^n)\\
  &\leq Pr\{(A^n,U^n(l),S^n_e)\notin (\mathcal{B}_2\cap \mathcal{T}_3) \;\;\text{for all $U^n(l)\in\textbf{C}$}|A^n=a^n,S^n_e=s^n_e\}.
\end{align*}
This term can be bounded following the same argument as the error event $\mathcal{E}_1$ in Section \ref{sec: proof of average distortion} where we set $R'=\bar{I}(\boldsymbol{U};\boldsymbol{S}_e|\boldsymbol{A}) + 2\gamma$ for some positive $\gamma>0$.
It follows that
\begin{align*}
  &Pr\{\frac{1}{n} d_n(S^n,g(f(M,S^n_e),f_{A}(M),S^n_e,Z^n))  > \bar{D}(\boldsymbol{S},\hat{\boldsymbol{S}}) + \gamma \}\\
  &\leq Pr\{\frac{1}{n} d_n(S^n,g(f(M,S^n_e),f_{A}(M),S^n_e,Z^n))  > \bar{D}(\boldsymbol{S},\hat{\boldsymbol{S}}) + \gamma | \hat{M} = M\}Pr\{\hat{M} = M\} +  P_e\\
  &\leq \pi_2^{\frac{1}{2}} + (\pi_2^{\frac{1}{2}} + \pi_3 + \exp(-\exp(n\gamma)))  +  P_e.
\end{align*}
With $n\to\infty$, we have $\pi_2\to 0, \pi_3 \to 0, \exp(-\exp(n\gamma)) \to 0$ and $P_e \to 0$ and hence,
\begin{align*}
  p-\limsup_{n\to\infty}\frac{1}{n} d_n(S^n,g(f(M,S^n_e),f_{A}(M),S^n_e,Z^n)) \leq \bar{D}(\boldsymbol{S},\hat{\boldsymbol{S}}) + \gamma \leq D+ \gamma.  
\end{align*}
The proof is completed. \hfill \qedsymbol

\section{Examples}\label{sec: numerical examples}
In this section, we apply our results to AWGN fading channel with additive interference. Although our general capacity-distortion formulas hold for arbitrary channel models, they are in general hard to compute. To derive some computable expressions, we simplify the channel model in the following two subsections. For the AWGN channel with additive interference, we investigate the boundary point for the capacity-distortion region with ergodic interference. It turns out that the distortion achieved in our result is also the minimum distortion that can be achieved when the interference is i.i.d. Gaussian. For the fading channel, we characterize the single-letter expression of the capacity-distortion formula for the ergodic fading process.
\subsection{AWGN with Additive Interference}
Let $\boldsymbol{A}=\boldsymbol{S}_d=\emptyset$. We first consider an AWGN channel with stationary and ergodic interference and compute the distortion that can be achieved when the user transmits a message with a capacity-achieving rate. The output of the channel is described by

\begin{align}
  \label{def: ergodic additive state model}Y_i = X_i + S_i + N_i, Z_i=X_i+S_i+N_{z,i}
\end{align}
where $N_{i}\sim\mathcal{N}(0,\sigma),N_{z,i}\sim\mathcal{N}(0,\sigma_z)$ with $\sigma_z > \sigma$,
and $\boldsymbol{S}=\{S_i\}_{i=1}^{\infty}$ is a stationary and ergodic source. We choose a linear MMSE estimator to reconstruct the state, i.e.
\begin{align*}
  \hat{S}_n = aZ + b X + c S_e
\end{align*}
for some real numbers $a,b,$ and $c$.
Further, the channel has an input constraint $\mathbb{E}\left[ X_n^2 \right]\leq P_X$. Assume the channel state information at the encoder side is a noisy version of the channel state:
\begin{align}
  \label{eq: definition of S_e}S_{e,i} = S_{i} + N_{e,i},
\end{align}
where $N_{e,i}\sim\mathcal{N}(0,\sigma_e)$. Now it is equivalent to considering the channel 
\begin{align*}
  Y_i = X_i + S_{e,i} + N_i - N_{e,i},\;\;i=1,2,3,\dots
\end{align*}
and the capacity
\begin{align}
  \label{eq: numerical example capacity}C(D) = \sup_{\bar{\mathcal{P}}_D}\ubar{I}(\boldsymbol{U};\boldsymbol{Y}) - \bar{I}(\boldsymbol{U};\boldsymbol{S}_e), 
\end{align}
where $\bar{\mathcal{P}}_D$ is the set of random processes in which each collection of random variables $(U^n,S^n_e,S^n,X^n,Y^n,Z^n)$ that satisfies $p-\limsup_{n\to\infty}\frac{1}{n}d_n(S^n,g_n(X^n,S^n_e,Z^n)) \leq D$.

Now substitute $U=X + \alpha S_e$ in \eqref{eq: numerical example capacity}, where $X\sim\mathcal{N}(0,P_X)$ is independent to $S$. It follows that
\begin{align*}
  &\ubar{I}(\boldsymbol{U};\boldsymbol{Y}) - \bar{I}(\boldsymbol{U};\boldsymbol{S}_e)\\
  &\geq \ubar{H}(\boldsymbol{U}) - \bar{H}(\boldsymbol{U}|\boldsymbol{Y}) - \bar{H}(\boldsymbol{U}) + \ubar{H}(\boldsymbol{U}|\boldsymbol{S}_e)\\
  &=\ubar{H}(\boldsymbol{U}|\boldsymbol{S}_e) - \bar{H}(\boldsymbol{U}|\boldsymbol{Y}) + \ubar{H}(\boldsymbol{U}) - \bar{H}(\boldsymbol{U}).
\end{align*}
Here we have
\begin{align*}
  \ubar{H}(\boldsymbol{U}|\boldsymbol{S}_e)=\ubar{H}(\boldsymbol{X}+\alpha\boldsymbol{S}_e|\boldsymbol{S}_e)=\ubar{H}(\boldsymbol{X}|\boldsymbol{S}_e)=p-\liminf_{n\to\infty} \frac{1}{n}\log \frac{1}{P_{X^n}(X^n)}
\end{align*}
and
\begin{align*}
  \bar{H}(\boldsymbol{U}|\boldsymbol{Y})&=\bar{H}(\boldsymbol{X}+\alpha\boldsymbol{S}_e|\boldsymbol{X}+\boldsymbol{S}_e + \boldsymbol{N}-\boldsymbol{N}_e)\\
  &=\bar{H}(\boldsymbol{X}-\alpha(\boldsymbol{X}+\boldsymbol{N}-\boldsymbol{N}_e)|\boldsymbol{X}+\boldsymbol{S}_e + \boldsymbol{N}-\boldsymbol{N}_e)\\
  &=p-\liminf_{n\to\infty} \frac{1}{n}\log \frac{1}{P(X^n-\alpha(X^n+N^n-N^n_e)|X^n+S^n_e+N^n-N^n_e)}
\end{align*}
Note that $X_i,N_i,N_{e,i}$ are all i.i.d. Gaussian random variables. By setting $\alpha=\frac{P_X}{\sigma + \sigma_e}$, $X_i-\alpha(X_i+N_i-N_{e,i})$ and $X_i+N_i-N_{e,i}$ are uncorrelated and jointly zero-mean Gaussian, and hence, independent\cite{cohen2002gaussian}. With the same argument as in \cite{cohen2002gaussian} we have that $X^n-\alpha(X^n+N^n-N^n_e)$ is independent of $X^n+S^n_e+N^n-N^n_e$ by the independence between $S^n_e$ and $(X^n,N^n,N^n_e)$. It follows that 
\begin{align*}
  &\lim_{n\to\infty}\frac{1}{n}\log \frac{1}{P_{X^n}(X^n)}=\lim_{n\to\infty}\frac{1}{n}\sum_{i=1}^n \log \frac{1}{P_{X}(X_i)}=h(X),\\
  &\lim_{n\to\infty} \frac{1}{n}\log \frac{1}{P(X^n-\alpha(X^n+N^n-N^n_e)|X^n+S^n_e+N^n-N^n_e)}\\
  &=\lim_{n\to\infty} \frac{1}{n}\log \frac{1}{P(X^n-\alpha(X^n+N^n-N^n_e))}\\
  &=\lim_{n\to\infty} \frac{1}{n}\sum_{i=1}^n \log \frac{1}{P(X_i-\alpha(X_i+N_i-N_{e,i}))}\\
  &=h(X|X+N-N_e).
\end{align*}
By an interleaving argument, we can treat the ergodic noise as an i.i.d. noise\cite{cohen2002gaussian}. Since $N_{e,i}$ and $X_i$ are i.i.d. Gaussian random variable, by Chebyshev's inequality we have
\begin{align*}
  \ubar{H}(\boldsymbol{U}) = \bar{H}(\boldsymbol{U})
\end{align*}
and the achieved rate is
\begin{align*}
  \frac{1}{2}\log(1+\frac{P_X}{\sigma+\sigma_e}).
\end{align*}
On the other hand, by the MMSE we have
\begin{align}
  \label{def: ergodic additive state linear mmse}a=\frac{\sigma_e\sigma_s}{\sigma_s\sigma_z+\sigma_e\sigma_s+\sigma_e\sigma_z},\;\;b=-a,\;\;c=\frac{\sigma_s\sigma_z}{\sigma_s\sigma_z+\sigma_e\sigma_s+\sigma_e\sigma_z},
\end{align}
and by the ergodicity of the state $\boldsymbol{S}$ we have
\begin{align}
  \label{eq: ergodic additive state distortion}d_n(\boldsymbol{S},\hat{\boldsymbol{S}})=\mathbb{E}[(S-\hat{S})^2]=\sigma_s - \frac{\sigma_s^2(\sigma_e+\sigma_z)}{\sigma_s\sigma_z+\sigma_e\sigma_s+\sigma_e\sigma_z}=\frac{\sigma_s\sigma_e\sigma_z}{\sigma_s\sigma_z+\sigma_e\sigma_s+\sigma_e\sigma_z}.
\end{align}

The above rate-distortion pair $(\frac{1}{2}\log(1+\frac{P_X}{\sigma+\sigma_e}),\frac{\sigma_s\sigma_e\sigma_z}{\sigma_s\sigma_z+\sigma_e\sigma_s+\sigma_e\sigma_z})$ is a boundary point of the rate-distortion region, where the message rate is maximized. Note that in the estimation phase, we can always eliminate input $X$ from the feedback $Z$ by the channel model defined in \eqref{def: ergodic additive state model} and the fact that the estimator has both the feedback $Z$ and the channel input $X$. Due to this the estimation is in fact performed based only on $N_z-N_e$ and $S_e=S+N_e$. We are interested in the minimum distortion that can be achieved with the assumption that the state source $\boldsymbol{S}$ is i.i.d. By the above argument, the estimation of state $S$ is equivalent to estimate $N_e$, and the distortion is equivalent to
\begin{align*}
  d_n(\boldsymbol{S},\hat{\boldsymbol{S}})=\mathbb{E}[(S-\hat{S})^2]=\mathbb{E}[(S-(S_e-\hat{N}_e))^2]=\mathbb{E}[(N_e-\hat{N}_e)^2].
\end{align*}
This means the choice of the distribution of input random variable does not affect the estimation phase and the distortion in \eqref{eq: ergodic additive state distortion} is still achievable in this case. To prove the upper bound, we follow the technique in \cite{sutivong2005channel}.
\begin{align*}
  &\frac{1}{n}I(N_e^n;N_z^n-N^n_e,S^n+N_e^n)\\
  &=\frac{1}{n}(h(N_e^n) - h(N_e^n|N_z^n-N^n_e,S^n+N_e^n))\\
  &\overset{(a)}{\geq}\frac{1}{n}(h(N_e^n) - h(N_e^n-\hat{N}_e^n(N_z^n-N^n_e,S^n+N_e^n)))\\
  &\geq\frac{1}{n}\sum_{i=1}^n (h(N_{e,i}) - h(N_{e,i}-\hat{N}_{e,i}))\\
  &\overset{(b)}{\geq}\frac{1}{n}\sum_{i=1}^n (\frac{1}{2}\log (2\pi e \sigma_e)- \frac{1}{2}\log (2\pi e \mathbb{E}[(N_{e,i}-\hat{N}_{e,i})^2]))\\
  &\geq \frac{1}{2}\log (2\pi e \sigma_e)- \frac{1}{2}\log (2\pi e \frac{1}{n}\sum_{i=1}^n \mathbb{E}[(N_{e,i}-\hat{N}_{e,i})^2]) = \frac{1}{2}\log \frac{\sigma_e}{D_n}
\end{align*}
where $\hat{N}_e$ in $(a)$ is a function of $(N_z^n-N^n_e,S^n+N_e^n)$, $(b)$ follows by the fact that Gaussian distribution maximizes the entropy for a given variance, and
\begin{align*}
  &\frac{1}{2}\log \frac{\sigma_e}{D_n}\leq \frac{1}{n}I(N_e^n;N_z^n-N^n_e,S^n+N_e^n)\\
  &=\frac{1}{n} (h(N_z^n-N^n_e,S^n+N_e^n)-h(N_z^n-N^n_e,S^n+N_e^n|N^n_e))\\
  &=\frac{1}{n} (h(N_z^n-N^n_e) + h(S^n+N_e^n|N_z^n-N^n_e)-h(N^n_z)-h(S^n))\\
  &=\frac{1}{n} (h(N_z^n-N^n_e) + h(S^n+N_e^n|N_z^n-N^n_e)-h(N^n_z)-h(S^n))\\
  &\leq \frac{1}{2}\log(2\pi e (\sigma_z + \sigma_e)) + \sum_{i=1}^n h(S_i+N_{e,i}|N_{z,i}-N_{e,i}) - \frac{1}{2}\log(2\pi e \sigma_z) - \frac{1}{2}\log(2\pi e \sigma_s)\\
  &\overset{(a)}{\leq} \frac{1}{2}\log(2\pi e (\sigma_z + \sigma_e)) + \sum_{i=1}^n h(S_i+N_{e,i}-a_{min}(N_{z,i}-N_{e,i})) - \frac{1}{2}\log(2\pi e \sigma_z) - \frac{1}{2}\log(2\pi e \sigma_s)\\
  &\leq \frac{1}{2}\log(2\pi e (\sigma_z + \sigma_e)) + \sum_{i=1}^n \frac{1}{2}\log (2\pi e \mathbb{E}[(S_i+N_{e,i}-a_{min}(N_{z,i}-N_{e,i}))^2]) - \frac{1}{2}\log(2\pi e \sigma_z) - \frac{1}{2}\log(2\pi e \sigma_s)\\
  &=\frac{1}{2}\log(2\pi e (\sigma_z + \sigma_e)) + \frac{1}{2}\log(2\pi e \frac{\sigma_s\sigma_z+\sigma_z\sigma_e+\sigma_s\sigma_e}{\sigma_z+\sigma_e}) - \frac{1}{2}\log(2\pi e \sigma_z) - \frac{1}{2}\log(2\pi e \sigma_s)\\
  &=\frac{1}{2}\log(\frac{\sigma_s\sigma_z+\sigma_z\sigma_e+\sigma_s\sigma_e}{\sigma_z\sigma_s})
\end{align*}
where in $(a)$, $a_{min}$ is the coefficient of linear MMSE estimator of $(S_i+N_{e,i})$ given $(N_{z,i}-N_{e,i})$ such that $a_{min}=\frac{N_e}{N+N_e}$. It follows that
\begin{align*}
  D_n \geq \frac{\sigma_z\sigma_s\sigma_e}{\sigma_s\sigma_z+\sigma_z\sigma_e+\sigma_s\sigma_e}.
\end{align*}
The above argument proves that when the state is an i.i.d. Gaussian source, the sender can transmit a message with maximal rate and minimum distortion. However, it is still an open problem if \eqref{eq: ergodic additive state distortion} is the minimum distortion that can be achieved for the ergodic state source since the linear MMSE estimator is optimal for jointly Gaussian random variables, but is not necessarily optimal for other cases. For the ergodic source, the estimator may use the correlation between the sequence $(S_1,S_2,\dots,S_n)$ to achieve a lower distortion. Hence, the minimum distortion for the ergodic state source is still an open problem.

For more general interference (non-ergodic), Costa's encoding is suboptimal. To see this, note that in this case the capacity is lower bounded by 
\begin{align*}
  &\ubar{I}(\boldsymbol{U};\boldsymbol{Y}) - \bar{I}(\boldsymbol{U};\boldsymbol{S}_e)\\
  &\geq \ubar{H}(\boldsymbol{U}) - \bar{H}(\boldsymbol{U}|\boldsymbol{Y}) - \bar{H}(\boldsymbol{U}) + \ubar{H}(\boldsymbol{U}|\boldsymbol{S}_e)\\
  &= \frac{1}{2}\log(1+\frac{P_X}{\sigma+\sigma_e}) + \ubar{H}(\boldsymbol{U})  - \bar{H}(\boldsymbol{U}),
\end{align*}
where $\ubar{H}(\boldsymbol{U})  - \bar{H}(\boldsymbol{U})$ is in general negative by the definitions of inf-entropy density and sup-entropy density. However, it was proved in \cite{erez2005capacity} that the interference can be canceled by the use of a lattice code with the help of common randomness $U^n$. By using a lattice code $\Lambda$ defined in \cite[Section IV]{erez2005capacity}, the model is equivalent to a modulo lattice additive noise channel described by 
\begin{align*}
  {Y'}^n = V^n + {N'}^n,
\end{align*}
where the noise ${N'}^n=[(1-\alpha)U^n+\alpha (N^n-N^n_e)] \mod \Lambda$ is independent to the input $V^n$ (the definition of operation $\mod \Lambda$ is given in \cite{erez2005capacity}), $\alpha=\frac{P_X}{P_X+\sigma+\sigma_e}$. Let the dither $U^n$ be uniformly distributed on its fundamental Voronoi region, the capacity of the channel
\begin{align*}
  \ubar{I}(\boldsymbol{V};\boldsymbol{Y})&= p-\liminf_{n\to\infty} \frac{1}{n}\log \frac{P_{Y^n|V^n}(Y^n|V^n)}{P_{Y^n}(Y^n)}\\
  &=p-\liminf_{n\to\infty} \frac{1}{n} \log \frac{1}{P_{Y^n}(Y^n)} - \frac{1}{n}\log \frac{1}{P_{{N'}^n}({N'}^n)}\\
  &\overset{(a)}{\geq} \frac{1}{2}\log \frac{P_X}{G(\Lambda)} - \frac{1}{2}\log(2\pi e \frac{P_X (\sigma+\sigma_e)}{P_X +\sigma+\sigma_e})\overset{(b)}{=}\frac{1}{2}\log (1+\frac{P_X}{\sigma+\sigma_e}),
\end{align*}
where $(a)$ follows by the fact that when $n\to\infty$, the dither $U^n$ can be regarded as a white Gaussian noise vector \cite{zamir1996lattice} and its normalized power is upper bounded by $\frac{P_X (\sigma+\sigma_e)}{P_X +\sigma+\sigma_e}$, $G(\Lambda)$ is the normalized second moment of $\Lambda$, $(b)$ follows by $G(\Lambda)\to\frac{1}{2\pi e}$ as $n\to\infty$\cite{zamir1996lattice}.

\subsection{AWGN with Fading}
In this section, we consider an AWGN channel with ergodic fading under average distoriton constraint described by 
\begin{align*}
  Y=SX+N, Z=SX+N_z,
\end{align*}
where $S$ is an ergodic source with limited power $\sigma_s$, $N\sim\mathcal{N}(0,\sigma),N_z\sim\mathcal{N}(0,\sigma_z)$ with $\sigma_z > \sigma$. We assume no side information at the encoder and perfect CSI at the decoder side. This is a generalization of the model considered in \cite{ahmadipour2022information}, and the capacity result for average distortion case is straightforward by setting $\boldsymbol{A}=\boldsymbol{S}_e=\emptyset$ and $\boldsymbol{S}_d=\boldsymbol{S}$ in Theorem \ref{the: average distortion capacity}. We define the state reconstruction function as in \cite{zhang2011joint} with
\begin{align*}
  \hat{S}(x,Z) = \frac{\sigma_s x}{\sigma_s|x|^2 + \sigma_z}Z,
\end{align*}
and the MMSE function given input symbol $x$ is
\begin{align*}
  \mathbb{E}[d(S,\hat{S}(x,Z))]=\mathbb{E}[|S-\hat{S}(x,Z)|^2] = \frac{\sigma_s\sigma_z}{\sigma_s|x|^2+\sigma_z}.
\end{align*}
The average distortion in this case is
\begin{align*}
  D(\boldsymbol{S},\hat{\boldsymbol{S}})&=\limsup_{n\to\infty}\frac{1}{n}\sum_{i=1}^n \mathbb{E}\left[ d(S_i,\hat{S}_i) \right]\\
  &=\limsup_{n\to\infty}\frac{1}{n}\sum_{i=1}^n \mathbb{E}\left[ (S_i-\hat{S}_i(X_i,Z_i))^2 \right]\\
  &=\limsup_{n\to\infty}\frac{1}{n}\sum_{i=1}^n  \mathbb{E}\left[ (S_i - \frac{\sigma_s X_i}{\sigma_s|X_i|^2 + \sigma_z}Z_i)^2 \right].
\end{align*}

\begin{corollary}
  The capacity-distortion region for an AWGN channel with stationary and ergodic fading is
\begin{align}
  \label{def: capacity distortion region of awgn with fading}C(D) = \max_{\substack{\alpha\in[0,1]:\\\mathbb{E}_{X\sim\mathcal{N}(0,\alpha P_X)}\left[ H(X) \right]\leq D}}\frac{1}{2}\mathbb{E}\left[ \log \frac{\alpha S^2 P_X+\sigma}{\sigma} \right],
\end{align}
where
\begin{align*}
  H(X) = \frac{\sigma_s\sigma_z}{\sigma_s|X|^2+\sigma_z}.
\end{align*} 
\end{corollary}

\begin{lemma}
  $C(D)$ is a non-decreasing and concave function of $D$.
\end{lemma}
The proof is similar to that of \cite[Lemma 2]{ahmadipour2022information} and is omitted here.
For the achievable part, set $X_i$ to be i.i.d. Gaussian random variable with zero mean and $\alpha P_X$ variance such that $\alpha\in[0:1]$ and $\mathbb{E}_{X}[H(X)]\leq D$. It follows that given state sequence $S^n$ at the decoder, the output $Y^n$ has i.i.d. Gaussian components with zero mean and variance $(\alpha S^2P_X+\sigma)$. Define a set of random processes $\mathcal{P}'(D):=\{\boldsymbol{X}:\limsup_{n\to\infty}\mathbb{E}[d_n(\boldsymbol{S},\hat{\boldsymbol{S}})]\leq D\}$. It follows that
\begin{align*}
  \sup_{\mathcal{P}'(D)} \underbar{I}(\boldsymbol{X};\boldsymbol{Y}|\boldsymbol{S})&\overset{(a)}{\geq} \liminf_{n\to\infty}\frac{1}{n}\log\frac{P_{Y^n|X^nS^n}(Y^n|X^n,S^n)}{P_{Y^n|S^n}(Y^n|S^n)}\\
  &=\liminf_{n\to\infty}\frac{1}{n}\sum_{i=1}^n \frac{1}{2}\left[ \log(\frac{\alpha S^2_iP_X+\sigma}{\sigma}) - \frac{|N|^2}{\sigma} + \frac{|Y_i|^n}{\alpha S^2_iP_X+\sigma}  \right]\\
  &\overset{(b)}{=}\frac{1}{2}\mathbb{E}\left[  \log(\frac{\alpha S^2P_X+\sigma}{\sigma}) \right]
\end{align*}
where $(a)$ follows by substituting i.i.d. Gaussian random variables $X_i\sim\mathcal{N}(0,\alpha P_X)$ in the inequality, $(b)$ follows by the fact that $\{X_i\}$ are i.i.d. and hence, $\{Y_i\}_{i=1,\dots,n}$ are stationary and ergodic.
The distortion is 
\begin{align*}
  D(\boldsymbol{S},\hat{\boldsymbol{S}})&\overset{(a)}{=} \mathbb{E}\left[ (S - \frac{\sigma_s X}{\sigma_s|X|^2 + \sigma_z}Z)^2 \right]\\
  &=\mathbb{E}\left[ \mathbb{E}\left[ (S - \frac{\sigma_s X}{\sigma_s|X|^2 + \sigma_z}Z)^2 \bigg| X\right]  \right]\\
  &=\mathbb{E}_{X\sim\mathcal{N}(0,\alpha P_X)}\left[ \frac{\sigma_s\sigma_z}{\sigma_s|X|^2+\sigma_z} \right]=\mathbb{E}_{X\sim\mathcal{N}(0,\alpha P_X)}\left[ H(X) \right]\leq D,
\end{align*}
where $(a)$ follows by the fact that $\boldsymbol{S}$ is stationary and ergodic having the same first-order distribution for all $S_i$ and $\{X_i\}$ are i.i.d random variables.

The upper bound follows by applying Theorem 3.5.2 in \cite{koga2013information},
\begin{align*}
  \underbar{I}(\boldsymbol{X};\boldsymbol{Y}|\boldsymbol{S}) &\leq \liminf_{n\to\infty}\frac{1}{n}I(X^n;Y^n|S^n)\\
  &= \liminf_{n\to\infty}\frac{1}{n} h(Y^n|S^n) - h(Y^n|X^n,S^n)\\
  &\leq \liminf_{n\to\infty}\frac{1}{n} \sum_{i=1}^n h(Y_i|S_i) - h(Y_i|X_i,S_i)\\
  &\overset{(a)}{= }\liminf_{n\to\infty}\frac{1}{n} \sum_{i=1}^n  \mathbb{E}\left[\frac{1}{2}\log 2\pi e(\alpha_i S^2_i P_{X_i} + \sigma) \right] - \frac{1}{2}\log 2\pi e\sigma\\
  &\overset{(b)}{\leq} \liminf_{n\to\infty}\frac{1}{n} \sum_{i=1}^n \mathbb{E}\left[ \frac{1}{2}\log \frac{\alpha_i S^2_iP_X+\sigma}{\sigma} \right]\\
  &\overset{(c)}{=}\liminf_{n\to\infty}\frac{1}{n} \sum_{i=1}^n \mathbb{E}\left[ \frac{1}{2}\log \frac{\alpha_i S^2P_X+\sigma}{\sigma} \right]\\
  &\overset{(d)}{\leq} \limsup_{n\to\infty}\frac{1}{n} \sum_{i=1}^n C(\mathbb{E}_{X_i\sim\mathcal{N}(0,\alpha_i P_X)}\left[ H(X_i) \right])\\
  &\overset{(e)}{\leq} \limsup_{n\to\infty} C(\frac{1}{n} \sum_{i=1}^n \mathbb{E}_{X_i\sim\mathcal{N}(0,\alpha_i P_X)}\left[ H(X_i) \right])\\
  &\overset{(f)}{=}C(\limsup_{n\to\infty} \frac{1}{n} \sum_{i=1}^n \mathbb{E}_{X_i\sim\mathcal{N}(0,\alpha_i P_X)}\left[ H(X_i) \right])=C(D(\boldsymbol{S},\hat{\boldsymbol{S}}))\overset{(g)}{\leq} C(D).
\end{align*}
where $(a)$ holds since the differential entropy given a variance is maximized for Gaussian distribution and $\frac{1}{2}\log 2\pi e\sigma \leq h(Y_i|S_i)\leq \mathbb{E}\left[\frac{1}{2}\log 2\pi e( S^2_i P_{X_i} + \sigma) \right]$, $(b)$ follows by the constraint $E[|X_i|^2]\leq P_X$ for all $i$, $(c)$ follows by the fact that $\boldsymbol{S}$ is stationary and ergodic and, hence, has the same first order distribution for all $S_i$, $(d)$ follows by the definition of $C(D)$ in \eqref{def: capacity distortion region of awgn with fading}, $(e)$ follows by the concavity of $C(D)$, $(f)$ follows by the fact that $C(D)$ is continuous and non-decreasing function of $D$, $(g)$ follows since $C(D)$ is non-decreasing.

For the case that the distribution is the mixture of two stationary and ergodic sources, i.e.
\begin{align*}
  P_{S^n} = \beta P_{S^n_1} + (1-\beta)P_{S^n_2},
\end{align*}
we have the following inner bound result.
\begin{corollary}
  The capacity-distortion region of AWGN with mixed fading satisfies
  \begin{align*}
     \max_{\substack{\alpha\in[0,1]:\\\mathbb{E}_{X\sim\mathcal{N}(0,\alpha P_X)}\left[ \widetilde{H}(X) \right]\leq D}}\min\left\{\frac{1}{2}\mathbb{E}\left[  \log(\frac{\alpha S^2_1P_X+\sigma}{\sigma}) \right],\frac{1}{2}\mathbb{E}\left[  \log(\frac{\alpha S^2_2P_X+\sigma}{\sigma}) \right]\right\}\subseteq C_{MIX}(D),
  \end{align*}
  where
  \begin{align*}
    \widetilde{H}(X) = \beta \frac{\sigma_{s_1}\sigma_z}{\sigma_{s_1}|X|^2+\sigma_z} + (1-\beta) \frac{\sigma_{s_2}\sigma_z}{\sigma_{s_2}|X|^2+\sigma_z}.
  \end{align*}
\end{corollary}

The channel input $\{X_i\}$ is set to be i.i.d. Gaussian random variables with zero mean and $\alpha P_X$ variance such that $\alpha\in[0,1]$ and
\begin{align*}
  \beta \mathbb{E}_{X\sim\mathcal{N}(0,\alpha P_X)}\left[ \frac{\sigma_{s_1}\sigma_z}{\sigma_{s_1}|X|^2+\sigma_z} \right] + (1-\beta) \mathbb{E}_{X\sim\mathcal{N}(0,\alpha P_X)}\left[ \frac{\sigma_{s_2}\sigma_z}{\sigma_{s_2}|X|^2+\sigma_z} \right] \leq D.
\end{align*}
It follows  that
\begin{align*}
  \underbar{I}(\boldsymbol{X};\boldsymbol{Y}|\boldsymbol{S}) &= \min\{I(\boldsymbol{X};\boldsymbol{Y}_1|\boldsymbol{S}_1),I(\boldsymbol{X};\boldsymbol{Y}_2|\boldsymbol{S}_2)\}\\
  &=\min\{\frac{1}{2}\mathbb{E}\left[  \log(\frac{\alpha S^2_1P_X+\sigma}{\sigma}) \right],\frac{1}{2}\mathbb{E}\left[  \log(\frac{\alpha S^2_2P_X+\sigma}{\sigma}) \right]\}
\end{align*}
For the distortion we have 
\begin{align*}
  D(\boldsymbol{S},\hat{\boldsymbol{S}})&=\limsup_{n\to\infty}\frac{1}{n}\mathbb{E}\left[ d(S^n,\hat{S}^n) \right]\\
  &=\limsup_{n\to\infty}(\frac{\beta}{n}\mathbb{E}\left[ d(S^n_1,\hat{S}^n_1) \right] + \frac{1-\beta}{n}\mathbb{E}\left[ d(S^n_2,\hat{S}^n_2) \right])\\
  &\overset{(a)}{=}\beta \mathbb{E}_{X\sim\mathcal{N}(0,\alpha P_X)}\left[ \frac{\sigma_{s_1}\sigma_z}{\sigma_{s_1}|X|^2+\sigma_z} \right] + (1-\beta) \mathbb{E}_{X\sim\mathcal{N}(0,\alpha P_X)}\left[ \frac{\sigma_{s_2}\sigma_z}{\sigma_{s_2}|X|^2+\sigma_z} \right],
\end{align*}
where $(a)$ follows by the fact that $\frac{1}{n}\mathbb{E}\left[ d(S^n_1,\hat{S}^n_1)\right]$ and $\frac{1}{n}\mathbb{E}\left[ d(S^n_2,\hat{S}^n_2)\right]$ converge since $\boldsymbol{S}_1$ and $\boldsymbol{S}_2$ are stationary and ergodic and $\{X_i\}$ are i.i.d. random variables.

\section{Conclusion}
In this paper, we provide the general formulas of action-dependent ISAC problems under different distortion criteria. Results of special cases include stationary and memoryless channels, mixed states and channels, and rate-limited side information at one side case. Numerical results focus on AWGN with interference and fading channels.

\appendices

\section{proof of corollary \ref{coro: discrete memoryless capacity}}\label{app: proof of discrete memoryless capacity}
In this section, we prove Corollary \ref{coro: discrete memoryless capacity}. We start by showing for discrete generated states and discrete channels, generating codewords in a memoryless way does not make the distortion larger.

Let $\hat{\boldsymbol{S}}=\{\hat{S}^n=\{\hat{S}^{(n)}_1,\hat{S}^{(n)}_2,\dots,\hat{S}^{(n)}_n\}\}_{n=1}^{\infty}$ be the general reproduction process of channel state such that we have $\hat{S}^n=g_n(X^n,A^n,S^n_e,Z^n)$ using reproduce function $g_n$.
Let $(\boldsymbol{A}^*,\boldsymbol{U}^*,\boldsymbol{X}^*)$ be random process in which each collection of random variables $(A^{n*},U^{n*},X^{n*})$ is distributed as
\begin{align}
  \label{def:memoryless A}P_{A^{n*}}(a^n) = \prod_{i=1}^{n} P_{A^{(n)*}_i}(a_i),\\
  \label{def:memoryless U}P_{U^{n*}|A^{n*}S^n_e}(u^n|a^n,s^n_e) = \prod_{i=1}^{n} P_{U^{(n)*}_i|A^{(n)*}_iS^{(n)}_{e,i}}(u_i|a_i,s_{e,i}),\\
  \label{def:memoryless X}P_{X^{n*}|U^{n*}S^n_e}(x^n|u^n,s^n_e) = \prod_{i=1}^{n} P_{X^{(n)*}_i|U^{(n)*}_iS^{(n)}_{e,i}}(x_i|u_i,s_{e,i}).
\end{align}
 and $\hat{\boldsymbol{S}}^*=\{\hat{S}^{n*}=(\hat{S}^{(n)*}_1,\hat{S}^{(n)*}_2,\dots,\hat{S}^{(n)*}_n)\}$ be the reproduction process induced by $(\boldsymbol{A}^*,\boldsymbol{U}^*,\boldsymbol{X}^*,\boldsymbol{S}^{*}_e,\boldsymbol{S}^{*},\boldsymbol{S}^{*}_d,\boldsymbol{Y}^{*},\boldsymbol{Z}^{*})$ such that $\hat{S}^{n*}=g_n(X^{n*},A^{n*},S^{n*}_e,Z^{n*})$ with each component satisfying
\begin{align*}
  P_{S^{(n)}_i\hat{S}^{(n)}_i} = P_{S^{(n)}_i\hat{S}^{(n)*}_i},\;\;\text{$i=1,2,\dots,n.$}
\end{align*}
Note that given reproduce function $g_n$, the reproduced sequence $\hat{S}^n$ is determined by $(X^n,A^n,S^n_e,Z^n)$ and $\hat{S}^{n*}$ follows likewise.
Further note that $(S^{(n)}_1,\hat{S}^{(n)*}_1),\dots,(S^{(n)}_n,\hat{S}^{(n)*}_n)$ are independent due to the discrete memoryless property of the random processes.  If the random processes $(\boldsymbol{A}^*,\boldsymbol{U}^*,\boldsymbol{X}^*)$ also satisfy that $P^{(n)}_{A_i}=P_A,P^{(n)}_{U_i|A_iS_{e,i}}=P_{UX|AS_{e}},P^{(n)}_{X_i|U_iS_{e,i}}=P_{X|US_e}$ for all $i=1,\dots,n$, we say they are stationary and memoryless processes and omit the superscript $`(n)'$. The proof of the following lemma is similar to that of \cite[Lemma 5.8.1]{koga2013information}.
\begin{lemma}\label{lem: distortion constraint for memoryless case}
  For action-dependent ISAC model with memoryless state and memoryless channel defined in \eqref{eq: discrete state condition} and \eqref{eq: discrete channel condition}, respectively and additive distortion defined in \eqref{def: additive distortion function}, we have
  \begin{align*}
    \bar{D}(\boldsymbol{S},\hat{\boldsymbol{S}}^*) \leq \bar{D}(\boldsymbol{S},\hat{\boldsymbol{S}}),\;\; D(\boldsymbol{S},\hat{\boldsymbol{S}}^*) = D(\boldsymbol{S},\hat{\boldsymbol{S}})
  \end{align*}
\end{lemma}
\begin{proof}
  By the assumption that $d_n$ is an additive distortion, we have
\begin{align*}
  \frac{1}{n}\mathbb{E}\left[ d_n(S^n,\hat{S}^n) \right] = \frac{1}{n}\mathbb{E}\left[ \sum_{i=1}^n d(S^{(n)}_i,\hat{S}^{(n)}_i) \right]\overset{(a)}{=} \frac{1}{n}\mathbb{E}\left[ \sum_{i=1}^n d(S^{(n)}_i,\hat{S}^{(n)*}_i) \right],
\end{align*}
where $(a)$ follows by the fact that $P_{S^{(n)}_i\hat{S}^{(n)}_i} = P_{S^{(n)}_i\hat{S}^{(n)*}_i}$. By noticing that $0\leq d(S,\hat{S})\leq D_{max} < \infty$, applying Chebyshev's inequality yields
\begin{align*}
  \bar{D}(\boldsymbol{S},\hat{\boldsymbol{S}}^*) &= p-\limsup_{n\to\infty}\frac{1}{n}d_n(S^n,\hat{S}^{n*})\\
  &=\limsup_{n\to\infty}\frac{1}{n}\sum_{i=1}^n \mathbb{E}\left[ d(S^{(n)}_i,\hat{S}^{(n)*}_i) \right]\\
  &=\limsup_{n\to\infty}\frac{1}{n}\sum_{i=1}^n \mathbb{E}\left[ d(S^{(n)}_i,\hat{S}^{(n)}_i) \right]\\
  &=\limsup_{n\to\infty}\frac{1}{n} \mathbb{E}\left[ d_n(S^n,\hat{S}^n) \right]\\
  &= D(\boldsymbol{S},\hat{\boldsymbol{S}}).
\end{align*}
Now by the uniform integrability condition and \cite[Lemma 5.3.2]{koga2013information}, it follows that $D(\boldsymbol{S},\hat{\boldsymbol{S}}) \leq \bar{D}(\boldsymbol{S},\hat{\boldsymbol{S}})$ and hence, $\bar{D}(\boldsymbol{S},\hat{\boldsymbol{S}}^*) \leq \bar{D}(\boldsymbol{S},\hat{\boldsymbol{S}}).$
For the average distortion, we have
\begin{align*}
  D(\boldsymbol{S},\hat{\boldsymbol{S}}) &= \limsup_{n\to\infty}\frac{1}{n} \mathbb{E}\left[ d(S^n,\hat{S}^n) \right]\\
  &= \limsup_{n\to\infty}\frac{1}{n} \sum_{i=1}^n\mathbb{E}\left[ d(S^{(n)}_i,\hat{S}^{(n)}_i) \right]\\
  &= \limsup_{n\to\infty}\frac{1}{n} \sum_{i=1}^n\mathbb{E}\left[ d(S^{(n)}_i,\hat{S}^{(n)*}_i) \right]\\
  &=\limsup_{n\to\infty}\frac{1}{n} \mathbb{E}\left[ d(S^{n},\hat{S}^{n*}) \right]\\
  &=D(\boldsymbol{S},\hat{\boldsymbol{S}}^*).
\end{align*}  
The proof is completed.
\end{proof}
Lemma \ref{lem: distortion constraint for memoryless case} shows that when the channel states and the channels are discrete memoryless, restricting the action sequences and codewords to be generated in discrete memoryless ways does not violate the distortion constraint. 
In the following subsections, we prove the maximal distortion case for the stationary and memoryless channel and state. The proof of the average distortion case follows similarly.
\subsection{Converse Part for Maximal Distortion Case}
Now, we prove the converse part. By \cite[Theorem 3.5.2]{koga2013information}, we have 
\begin{align*}
  \ubar{I}(\boldsymbol{A},\boldsymbol{U};\boldsymbol{Y}) \leq \liminf_{n\to\infty} \frac{1}{n}I(A^n,U^n;Y^n),\\
  \bar{I}(\boldsymbol{U};\boldsymbol{S}|\boldsymbol{A}) \geq \limsup_{n\to\infty} \frac{1}{n}I(U^n;S^n_e|A^n).
\end{align*}
It follows that
\begin{align*}
  &\ubar{I}(\boldsymbol{A},\boldsymbol{U};\boldsymbol{Y}) - \bar{I}(\boldsymbol{U};\boldsymbol{S}|\boldsymbol{A})\\
  &\leq\liminf_{n\to\infty} \frac{1}{n}I(A^n,U^n;Y^n) - \limsup_{n\to\infty} \frac{1}{n}I(U^n;S^n_e|A^n)\\
  &=\liminf_{n\to\infty} \frac{1}{n} \left[  I(A^n,U^n;Y^n) - I(U^n;S^n_e|A^n) \right],
\end{align*}
where
\begin{align}
  &\frac{1}{n}I(A^n,U^n;Y^n) - I(U^n;S^n_e|A^n)\notag\\
  &=\frac{1}{n}\sum_{i=1}^n I(A^n,U^n;Y_i|Y^{i-1}) - I(U^n;S_{e,i}|S_{e,i+1}^{n},A^n)\notag\\
  &\overset{(a)}{=}\frac{1}{n}\sum_{i=1}^nI(A^n,U^n,S_{e,i+1}^{n};Y_i|Y^{,i-1}) - I(U^n,Y^{i-1};S_{e,i}|S_{e,i+1}^{n},A^n)\notag\\
  &\leq \frac{1}{n}\sum_{i=1}^nI(A^n,U^n,S_{e,i+1}^{n},Y^{i-1};Y_i) - (H(S_{e,i}|S_{e,i+1}^{n},A^n) - H(S_{e,i}|S_{e,i+1}^{n},A^n,U^n,Y^{i-1}))\notag\\
  &\label{eq: nonstationary and memoryless 1}\overset{(b)}{=} \frac{1}{n}\sum_{i=1}^nI(U^{*}_i,A_i;Y_i) - I(S_{e,i};U^{*}_i|A_i)\\
  &=\frac{1}{n}\sum_{i=1}^nI(U^{*}_{Q_n},A_{Q_n};Y_{Q_n}|Q_n=i) - I(S_{e,{Q_n}};U^{*}_{Q_n}|A_{Q_n},Q_n=i)\notag\\
  &=I(U^{*}_{Q_n},A_{Q_n};Y_{Q_n}|Q_n) - I(S_{e,{Q_n}};U^{*}_{Q_n}|A_{Q_n},Q_n)\notag\\
  &\overset{(c)}{\leq} I(U^{*}_{Q_n},A_{Q_n},Q_n;Y_{Q_n}) - I(S_{e,{Q_n}};U^{*}_{Q_n},Q_n|A_{Q_n})\notag\\
  \label{neq: converse upper bound}&\overset{(d)}{=} I(U_n,A_n;Y_n) - I(U_n;S_{e,n}|A_n)
\end{align}
where $(a)$ follows by applying Csisz\'ar's sum identity, $(b)$ follows by setting $U^{*}_i=(A^n,U^n,S_{e,i+1}^{n},Y^{i-1})$ and the memoryless assumption, $(c)$ follows by the independence between $S_{e,Q_n}$ and $Q_n$ given $A_{Q_n}$ and $(d)$ follows by setting $U_n = (U^{*}_{Q_n},Q_n),A_n=A_{Q_n}, S_{e,n} = S_{e,{Q_n}}, Y_n = Y_{Q_n}$. We further define $X_n = X_{Q_n},S_n=S_{Q_n},Z_n = Z_{Q_n}$. Note that for all distributions $P_{S^n_eS^n|U^n}$ and $P_{Y^nZ^n|X^nS^n}$ with memoryless property, the random variables in equation \eqref{neq: converse upper bound} also satisfy the Markov chain relations $(A_n,U_n,S_{e,n},S_{d,n})-(X_n,S_n)-(Y_n,Z_n)$ and $X_n-(U_n,S_{e,n})-A_n$. Hence, for $P_{A^n}=\prod_{i=1}^{n}P_{A_i}$ and $P_{U^n|A^nS^n_e}P_{X^n|U^nS^n_e}=\prod_{i=1}^{n}P_{U_i|A_iS_{e,i}}P_{X_i|U_iS_{e,i}}$, inequality \eqref{neq: converse upper bound} still holds. 

Let $g_{n,i}(A^n,X^n,S^n_e,Z^n) := \hat{S}^{n}_i$ be the $i-$th component of $\hat{S}^n.$ Further note that when $(A^n,U^n,X^n)$ are generated in a memoryless way, we have the best state estimator as defined in \eqref{def: memoryless best state estimator}, denoted by $g^*$. From Lemma \ref{lem: distortion constraint for memoryless case}, we know that restricting processes $(\boldsymbol{A},\boldsymbol{U},\boldsymbol{X})$ to memoryless processes does not violate the distortion constraint.
By the definition of $\bar{D}$ we have 
\begin{align*}
  \bar{D}(\boldsymbol{S},\hat{\boldsymbol{S}})&=\limsup_{n\to\infty}\frac{1}{n}\sum_{i=1}^n\mathbb{E}\left[ d\left(S_i,\hat{S}_i\right) \right]\\
  &=\limsup_{n\to\infty}\frac{1}{n}\sum_{i=1}^n \mathbb{E}\left[ d\left(S_i,g_{n,i}(A^n,X^n,S^n_e,Z^n)\right) \right]\\
  &\overset{(a)}{\geq} \limsup_{n\to\infty}\frac{1}{n}\sum_{i=1}^n \mathbb{E}\left[ d\left(S_i,g^*(A_i,X_i,S_{e,i},Z_i)\right) \right]\\
  &=\limsup_{n\to\infty}\mathbb{E}\left[ \mathbb{E}\left[ d \left( S_{Q_n},g^*(A_{Q_n},X_{Q_n},S_{e,{Q_n}},Z_{Q_n}) \right) \right] | Q_n \right]\\
  &\overset{(b)}{=}\limsup_{n\to\infty} \mathbb{E} \left[ d\left(S_n,g^*(A_n,X_n,S_{e,n},Z_n)\right) \right]\\
  &\overset{(c)}{=}\limsup_{n\to\infty} \mathbb{E} \left[ d\left(S_n,\hat{S}_n\right) \right],
\end{align*}
where $(a)$ follows by replacing the estimator $g_n$ with the best estimator $g^*$, $(b)$ follows by setting $A_n=A_{Q_n}, S_{e,n} = S_{e,{Q_n}}, X_n = X_{Q_n},Z_n = Z_{Q_n}$ as in \eqref{neq: converse upper bound}, $(c)$ follows by setting $\hat{S}_n:=g^*(A_n,X_n,S_{e,n},Z_n)$.
Since Lemma \ref{lem: distortion constraint for memoryless case} implies that restricting input distortion to be memoryless does not violate the constraint, we assume $\bar{D}(\boldsymbol{S},\hat{\boldsymbol{S}}) \leq D$ and hence,
\begin{align}
  \label{neq: single letter distortion constraint}\mathbb{E}\left[ d(S_n,\hat{S}_n) \right] \leq D+ \gamma
\end{align}
for some $\gamma>0$ when $n>n_0$ for some sufficiently large $n_0$. Combining \eqref{neq: converse upper bound} and \eqref{neq: single letter distortion constraint} yields
\begin{align*}
  \ubar{I}(\boldsymbol{A},\boldsymbol{U};\boldsymbol{Y}) - \bar{I}(\boldsymbol{U};\boldsymbol{S}|\boldsymbol{A}) \leq \max_{\substack{P_A,P_{U|AS_e}P_{X|US_e}:\\\mathbb{E}\left[ d(S,\hat{S}) \right] \leq D+ \gamma}} I(A,U;Y) - I(U;S_e|A).
\end{align*}
Since the right-hand side formula is a continuous function with respect to $\gamma,$ by letting $\gamma\to0$, we have
\begin{align*}
  \ubar{I}(\boldsymbol{A},\boldsymbol{U};\boldsymbol{Y}) - \bar{I}(\boldsymbol{U};\boldsymbol{S}_e|\boldsymbol{A}) \leq \max_{\substack{P_A,P_{U|AS_e}P_{X|US_e}:\\\mathbb{E}\left[ d(S_n,\hat{S}_n) \right] \leq D}} I(A,U;Y) - I(U;S_e|A).
\end{align*}

\subsection{Achievability for Maximal Distortion Case}
Let $\bar{\mathcal{P}}_{D,M}$ be the set of random processes in which each collection of random variables $(A^n,U^n,S^n_e,S^n,X^n,Y^n,Z^n)$ satisfying $p-\limsup_{n\to\infty}\frac{1}{n}d_n(S^n,g(X^n,A^n,S^n_e,Z^n)) \leq D$ and $(\boldsymbol{A},\boldsymbol{U},\boldsymbol{X})$ are random processes satisfying memoryless condition in \eqref{def:memoryless A} - \eqref{def:memoryless X}. It follows that
\begin{align*}
  C(D) &\geq \sup_{\bar{\mathcal{P}}_{D,M}} \ubar{I}(\boldsymbol{A},\boldsymbol{U};\boldsymbol{Y}) - \bar{I}(\boldsymbol{U};\boldsymbol{S}_e|\boldsymbol{A})\\
  &=\sup_{\bar{\mathcal{P}}_{D,M}} p-\liminf \frac{1}{n}\log \frac{P_{Y^n|A^nU^n}(Y^n|A^n,U^n)}{P_{Y^n}(Y^n)} - p-\limsup \frac{1}{n}\log \frac{P_{S^n_e|A^nU^n}(S^n_e|A^nU^n)}{P_{S^n_e|A^n}(S^n_e|A^n)}\\
  &\overset{(a)}{\geq} \liminf_{n\to\infty}\frac{1}{n}\log \frac{P_{Y^{n*}|A^{n*}U^{n*}}(Y^{n*}|A^{n*},U^{n*})}{P_{Y^{n*}}(Y^{n*})} - \limsup_{n\to\infty} \frac{1}{n}\log \frac{P_{S^{n*}_e|A^{n*}U^{n*}}(S^{n*}_e|A^{n*}U^{n*})}{P_{S^{n*}_e|A^{n*}}(S^{n*}_e|A^{n*})}\\
  &\overset{(b)}{=}\mathbb{E}\left[ \frac{P_{Y^{*}|A^{*}U^{*}}(Y^{*}|A^{*},U^{*})}{P_{Y^{*}}(Y^{*})}   \right] - \mathbb{E}\left[  \log \frac{P_{S^{*}_e|A^{*}U^{*}}(S^{*}_e|A^{*}U^{*})}{P_{S^{*}_e|A^{*}}(S^{*}_e|A^{*})} \right]\\
  &= I(A^*,U^*;Y^*) - I(U^*;S^*_e|A^*),
\end{align*}
where $(a)$ follows by substituting $(A^*,U^*,S^*_e,Y^*)$ achieving maximum in \eqref{eq: stationary and memoryless capacity 1} into the formula, $(b)$ follows by the stationary and memoryless properties and applying Chebyshev's inequality. For the maximal distortion we have
\begin{align*}
  \bar{D}(\boldsymbol{S},\hat{\boldsymbol{S}}^*) = \mathbb{E}\left[ d(S,\hat{S}^*) \right]=\mathbb{E}\left[ d(S,g^*(A^*,X^*,S_e^*,Z^*)) \right] \leq D.
\end{align*}
The proof is completed.

\subsection{Direct and Converse Part of Average Distortion Case}
By Lemma \ref{lem: distortion constraint for memoryless case} we have
\begin{align*}
  D(\boldsymbol{S},\hat{\boldsymbol{S}}) =  D(\boldsymbol{S},\hat{\boldsymbol{S}}^*).
\end{align*}
The remaining direct and converse proof follows exactly the same as the maximal distortion constraint case.

\subsection{Nonstationary and Memoryless Case}
For the case that the states and channels are memoryless but nonstationary, we first prove the converse part. Upon having \eqref{eq: nonstationary and memoryless 1}, we have
\begin{align*}
  \ubar{I}(\boldsymbol{A},\boldsymbol{U};\boldsymbol{Y}) - \bar{I}(\boldsymbol{U};\boldsymbol{S}|\boldsymbol{A})\leq \frac{1}{n}\sum_{i=1}^nI(U^{(n)*}_i,A^{(n)}_i;Y^{(n)}_i) - I(S^{(n)}_{e,i};U^{(n)*}_i|A^{(n)}_i).
\end{align*}

 Note that Lemma \ref{lem: distortion constraint for memoryless case} only requires the memoryless property. For the maximal and average distortion in this case, we have
\begin{align*}
  D\geq D(\boldsymbol{S},\hat{\boldsymbol{S}}^*)&=\bar{D}(\boldsymbol{S},\hat{\boldsymbol{S}}^*)\\
  &=\limsup_{n\to\infty}\frac{1}{n}\sum_{i=1}^n \mathbb{E}\left[d(S_i,g_i(A^n,X^n,S^n_{e},Z^n)) \right]\\
  &\geq \limsup_{n\to\infty}\frac{1}{n}\sum_{i=1}^n \mathbb{E}\left[ d(S_i,g^*_i(A^{(n)}_i,X^{(n)}_i,S^{(n)}_{e,i},Z^{(n)}_i)) \right].
\end{align*}
This gives the upper bound for the nonstationary and memoryless case. The achievability part is similar to that in \cite[Appendix D]{tan2014formula} with an additional distortion constraint. Note that by Lemma \ref{lem: distortion constraint for memoryless case}, restricting input distributions to be memoryless does not violate the constraint and hence, $\mathcal{P}^n_{NM,D}$ is not an empty set. This completes the proof.

\section{proof of theorem \ref{the: mixed maximal distortion} and theorem \ref{the: mixed average distortion}}\label{app: proof of mixed case results}
In this section, we prove the capacity-distortion results about mixed states and channels in Theorems \ref{the: mixed maximal distortion} and \ref{the: mixed average distortion}.

For given input distribution $P_{A^n}$ and $P_{U^nX^n|A^nS^n_e}$, the joint distribution of $(A^n,U^n,X^n,S^n_e,S^n,S^n_d,Y^n,Z^n)$ is
\begin{align*}
  &P_{A^n,U^n,X^n,S^n_e,S^n,S^n_d,Y^n,Z^n}(a^n,u^n,x^n,s^n_e,s^n,s^n_d,y^n,z^n)\\
  &=P_{A^n}(a^n)\left[ \alpha_1 P_{S^n_{e,1}S^n_1S^n_{d,1}|A^n}(s^n_e,s^n,s^n_d|a^n) + \alpha_2P_{S^n_{e,2}S^n_2S^n_{d,2}|A^n}(s^n_e,s^n,s^n_d|a^n) \right]  P_{U^n|A^nS^n_e}(u^n|a^n,s^n_e)P_{X^n|U^nS^n_e}(x^n|u^n,s^n_e)\\
  &\quad\quad\quad\quad \cdot \left[ \beta_1 P_{Y^n_1Z^n_1|X^nS^n}(y^n,z^n|x^n,s^n) + \beta_2 P_{Y^n_2Z^n_2|X^nS^n}(y^n,z^n|x^n,s^n) \right]\\
  &=\sum_{i=1,2}\sum_{j=1,2}\alpha_i\beta_jP_{A^n}(a^n)P_{S^n_{e,i}S^n_iS^n_{d,i}|A^n}(s^n_e,s^n,s^n_d|a^n)P_{U^n|A^nS^n_{e,i}}(u^n|a^n,s^n_e)P_{X^n|U^nS^n_{e,i}}(x^n|u^n,s^n_e)P_{Y^n_{i,j}Z^n_{i,j}|X^nS^n_i}(y^n,z^n|x^n,s^n)\\
  &=\sum_{i=1,2}\sum_{j=1,2}\alpha_i\beta_j P_{A^n,U^n,X^n,S^n_{e,i},S^n_i,S^n_{d,i},Y^n_{i,j},Z^n_{i,j}}(a^n,u^n,x^n,s^n_e,s^n,s^n_d,y^n,z^n),
\end{align*}
and the marginal distributions of $(A^n,U^n,S^n_e)$ and $(A^n,U^n,S^n_d,Y^n)$ are
\begin{align*}
  P_{A^n,U^n,S^n_e}(a^n,u^n,s^n_e) = \sum_{i=1,2}\alpha_i P_{A^n,U^n_i,S^n_{e,i}}(a^n,u^n,s^n_e),\\
  P_{A^n,U^n,S^n_d,Y^n}(a^n,u^n,s^n_d,y^n) = \sum_{i=1,2}\sum_{j=1,2}\alpha_i\beta_j P_{A^n,U^n_i,S^n_{d,i},Y^n_{i,j}}(a^n,u^n,s^n_d,y^n).
\end{align*}
\subsection{Proof of Theorem \ref{the: mixed maximal distortion}}
Let $\hat{S}^n_{i,j}=g_n(A^n,X^n,S^n_{e,i},Z^n_{i,j})$. By \cite[Lemma 1.4.2]{koga2013information}, which can be easily extended to finite convex combination case, we have
\begin{align}
  \bar{D}(\boldsymbol{S},\hat{\boldsymbol{S}}) = \max\left\{ \bar{D}(\boldsymbol{S}_1,\hat{\boldsymbol{S}}_{1,1}), \bar{D}(\boldsymbol{S}_1,\hat{\boldsymbol{S}}_{1,2}), \bar{D}(\boldsymbol{S}_2,\hat{\boldsymbol{S}}_{2,1}),\bar{D}(\boldsymbol{S}_2,\hat{\boldsymbol{S}}_{2,2}),\right\}.
\end{align}
In addition, \cite[Lemma 3.3.1]{koga2013information} and \cite[Lemma 5.10.1]{koga2013information} imply that
\begin{align*}
  \ubar{I}(\boldsymbol{A},\boldsymbol{U};\boldsymbol{S}_d,\boldsymbol{Y}) = \min\left\{ \ubar{I}(\boldsymbol{A},\boldsymbol{U}_1;\boldsymbol{S}_{d,1},\boldsymbol{Y}_{1,1}),\ubar{I}(\boldsymbol{A},\boldsymbol{U}_1;\boldsymbol{S}_{d,1},\boldsymbol{Y}_{1,2}),
  \ubar{I}(\boldsymbol{A},\boldsymbol{U}_2;\boldsymbol{S}_{d,2},\boldsymbol{Y}_{2,1}),
  \ubar{I}(\boldsymbol{A},\boldsymbol{U}_2;\boldsymbol{S}_{d,2},\boldsymbol{Y}_{2,2}) \right\},\\
  \bar{I}(\boldsymbol{U};\boldsymbol{S}_e|\boldsymbol{A}) = \max\left\{ \bar{I}(\boldsymbol{U};\boldsymbol{S}_{e,1}|\boldsymbol{A}),\bar{I}(\boldsymbol{U};\boldsymbol{S}_{e,2}|\boldsymbol{A}) \right\}.
\end{align*}
Substituting the above terms into Theorem \ref{the: maximal distortion capacity} completes the first part proof of Theorem \ref{the: mixed maximal distortion}.
For memoryless and stationary cases, the achievability  follows similarly to \cite{tan2014formula} and \cite{bloch2008secrecy} by considering the optimal distribution in \eqref{def: rate of mixed memoryless} and define random variables $(S^n_{e,i},S^n_i,S^n_{d,i},U^n,X^n,Y^n_{i,j},Z^n_{i,j})$ by $\prod_{l=1}^{n}P_A(a_l)P_{S_{e,i}S_iS_{d,i}|A}(s_{e,l},s_l,s_{d,l}|a_l)P_{U|AS_e}(u_l|s_{e,l},a_l)$ $P_{X|US_e}(x_l|s_{e,l},u_l)P_{Y_{j}Z_j|XS_i}(y_l,z_l|x_l,s_l)$. By the law of large numbers it follows that
\begin{align*}
  \ubar{I}(\boldsymbol{A},\boldsymbol{U}_i;\boldsymbol{S}_{d,i},\boldsymbol{Y}_{i,j}) = I(A,U;S_{d,i},Y_{i,j}),\;\;\; \text{$\forall i,j\in\{1,2\}$}\\
  \bar{I}(\boldsymbol{U}_i;\boldsymbol{S}_{e,i}|\boldsymbol{A}) = I(U;S_{e,i}|A),\;\;\; \text{$\forall i,j\in\{1,2\}$}.
\end{align*}
The proof is completed.

\subsection{Proof of Theorem \ref{the: mixed average distortion}}

For this case, we first rewrite the capacity-distortion function into distortion-capacity form,
\begin{align*}
  D(C) = \sup D(\boldsymbol{S},\hat{\boldsymbol{S}}),
\end{align*}
where $\sup$ is taken over random processes $(\boldsymbol{A},\boldsymbol{U},\boldsymbol{X},\boldsymbol{S}_{e},\boldsymbol{S},\boldsymbol{S}_{d},\boldsymbol{Y},\boldsymbol{Z})$ such that $\ubar{I}(\boldsymbol{A},\boldsymbol{U};\boldsymbol{Y},\boldsymbol{S}_d) - \bar{I}(\boldsymbol{U};\boldsymbol{S}_e|\boldsymbol{A})\leq C$ with $\boldsymbol{S}_e,\boldsymbol{S},\boldsymbol{S}_d$ having distribution given $\boldsymbol{A}$ as $P_{S^n_eS^nS^n_d|A^n}$ and $(\boldsymbol{Y},\boldsymbol{Z})$ given $(\boldsymbol{X},\boldsymbol{S})$ being distributed as $P_{Y^nZ^n|X^nS^n}$. It follows that
\begin{align*}
  D(\boldsymbol{S},\hat{\boldsymbol{S}}) &= \frac{1}{n}\limsup_{n\to\infty}\mathbb{E}\left[ d_n(S^n,\hat{S}^n) \right]\\
  &=\limsup_{n\to\infty}\mathbb{E}\left[ \sum_{i=1}^2\sum_{j=1}^2\frac{\alpha_i\beta_j}{n}d_n(S^n_i,\hat{S}^n_{i,j}) \right]\\
  &\leq \sum_{i=1}^2\sum_{j=1}^2 \alpha_i\beta_j \limsup_{n\to\infty}\frac{1}{n}\mathbb{E}\left[ d_n(S^n_i,\hat{S}^n_{i,j}) \right]\\
  &=\sum_{i=1}^2\sum_{j=1}^2 \alpha_i\beta_j D(\boldsymbol{S}_i,\hat{\boldsymbol{S}}_{i,j})
\end{align*}
By \cite[Lemma 3.3.1]{koga2013information} and \cite[Lemma 5.10.1]{koga2013information}, we have
\begin{align*}
  \ubar{I}(\boldsymbol{A},\boldsymbol{U};\boldsymbol{Y},\boldsymbol{S}_d) - \bar{I}(\boldsymbol{U};\boldsymbol{S}_e|\boldsymbol{A}) = \min_{i\in\{1,2\},j\in\{1,2\}}\ubar{I}(\boldsymbol{A},\boldsymbol{U}_i;\boldsymbol{S}_{d,i},\boldsymbol{Y}_{i,j}) - \max_{i\in\{1,2\}}\bar{I}(\boldsymbol{U}_i;\boldsymbol{S}_{e,i}|\boldsymbol{A}).
\end{align*}
Note that by the intersection between terms $D(\boldsymbol{S}_i,\hat{\boldsymbol{S}}_{i,j})$ for $i,j \in\{1,2\}$, we cannot take the supremum over each term separately.

\section{proof of theorem \ref{the: rate-limited CSI at encoder}}\label{app: proof of rate-limited csi at encoder}
In this section, we prove the capacity-distortion results for the case that there is rate-limited channel state information at the encoder side and imperfect CSI at the decoder side. In this case, the lossy description of the imperfect CSI can be regarded as common information at both the encoder and decoder sides. In addition, the binning is no longer necessary at the encoder side.

Given input random variables $(A^n,U^n,X^n,V^n,S^n,S^n_d,Y^n)$ with joint distribution $P_{A^n}P_{S^nS^n_d|A^n}P_{V^n|S^n_d}P_{X^n|A^nV^n}P_{Y^n|X^nS^n}$ such that $\frac{1}{n}\mathbb{E}\left[ d_n(S^n,g(X^n,A^n,V^n,Z^n)) \right] \leq D$.

Define mappings $F_n^{\mathcal{C}}:\mathcal{S}^n_d \to \mathcal{V}^n$
and  $\eta_1: \mathcal{U}^n \times \mathcal{A}^n \times \mathcal{V}^n \to \mathbb{R}^{+}$ and $\eta_2: \mathcal{U}^n \times \mathcal{A}^n \times \mathcal{V}^n \to \mathbb{R}^{+}$ as
\begin{align*}
  \eta_1(s^n_d,v^n) :&= \sum_{a^n,x^n,s^n,z^n} \sum_{\substack{y^n:\\(a^n,v^n,x^n,s^n_d,y^n)\notin\mathcal{T}_1}} P_{A^n|S^n_d}(a^n|s^n_d)P_{S^n|A^nS^n_d}(s^n|a^n,s^n_d)P_{X^n|A^nV^n}(x^n|a^n,v^n)P_{Y^nZ^n|X^nS^n}(y^n,z^n|x^n,s^n),\\
  \eta_2(s^n_d,v^n) :&= \sum_{a^n,x^n,s^n,y^n,z^n} P_{A^n|S^n_d}(a^n|s^n_d)P_{S^n|A^nS^n_d}(s^n|a^n,s^n_d)P_{X^n|A^nV^n}(x^n|a^n,v^n)P_{Y^nZ^n|X^nS^n}(y^n,z^n|x^n,s^n)d(s^n,g(a^n,v^n,x^n,z^n)),
\end{align*}
where
\begin{align*}
  \mathcal{T}_1 = \left\{(x^n,v^n,a^n,s^n_d,y^n):\frac{1}{n}\log \frac{P_{Y^nS^n_d|X^nV^nA^n}(y^n,s^n_d|x^n,v^n,a^n)}{P_{Y^nS^n_d|V^n}(y^n,s^n_d|v^n)} \geq \ubar{I}(\boldsymbol{A},\boldsymbol{X};\boldsymbol{Y},\boldsymbol{S}_d|\boldsymbol{V}) - \gamma\right\}.
\end{align*}
Further, define
\begin{align*}
  \mathcal{T}_2 = \left\{ (v^n,s^n_d): \frac{1}{n}\log\frac{P_{V^n|S^n_d}(v^n|s^n_d)}{P_{V^n}(v^n)} \leq \bar{I}(\boldsymbol{V};\boldsymbol{S}_d) + \gamma \right\}.
\end{align*}
and 
\begin{align*}
  \mathcal{T}_3 = \left\{ (v^n,a^n): \frac{1}{n}\log\frac{P_{V^n|A^n}(v^n|a^n)}{P_{V^n}(v^n)} \leq \bar{I}(\boldsymbol{V};\boldsymbol{A}) + \gamma \right\},\\
  \mathcal{B}=\{(s^n_d,v^n):\eta_1(s^n_d,v^n) \leq \pi_1^{\frac{1}{2}}\}.
\end{align*}
Similarly, we have 
\begin{align*}
  Pr\{(A^n,S^n,S^n_d,V^n,X^n,Y^n)\notin \mathcal{T}_1\}=\pi_1\to 0,\\
  Pr\{(V^n,S^n_d)\notin \mathcal{T}_2\}=\pi_2 \to 0,\\
  Pr\{(V^n,A^n)\notin \mathcal{T}_3\}=\pi_3 \to 0
\end{align*}
as $n\to \infty$ and $Pr\{(S^n_d,V^n)\notin \mathcal{B}\}\leq \pi_1^{\frac{1}{2}}$.

\emph{Codebook Generation.} Generate action codebook $\mathcal{A}=\{a^n(m):m\in[1:2^{nR}]\}$ according to $P_{A^n}$. Generate a set of lossy descriptions of the imperfect side information $\boldsymbol{S}_d$ at the decoder side $\mathcal{C}_e=\{v^n(l_e):l_e\in[1:2^{nR_e}]\}$, where $R_e = \bar{I}(\boldsymbol{V};\boldsymbol{S}_d)+2\gamma$, each according to distribution $P_{V^n}$.  For each $v^n(l_e)$, generate a message codebook $\mathcal{C}(l_e)=\{x^n(l_e,m):m\in[1:2^{nR}]\}$ with $R=\ubar{I}(\boldsymbol{A},\boldsymbol{X};\boldsymbol{Y},\boldsymbol{S}_d|\boldsymbol{V}) - 2\gamma$, each according to distribution $P_{X^n|A^nV^n}(\cdot | a^n(l_e,m),v^n(l_e))$. The codebooks are revealed to all the participants in the system.

\emph{Coded Side Information.} Once the message $m$ and action sequence $a^n(m)$ are determined, the state sequence $s^n$ and imperfect side information $s^n_d$ are generated. The state information encoder chooses $v^n(l_e)\in\mathcal{C}_e$ such that $(v^n(l_e),s^n_d)\in \mathcal{B}$.
If there is more than one such sequence, choose
\begin{align*}
  l_e^* = \mathop{\arg\min}_{\substack{l_e:v^n(l_e)\in\mathcal{C}_e,\\(v^n(l_e),s^n_d)\in\mathcal{B}}} \eta_2(v^n(l_e),s^n_d).
\end{align*}
If no such index exists, set
\begin{align*}
  l_e^* = \mathop{\arg\min}_{v^n(l_e)\in\mathcal{C}_e} \eta_2(v^n(l_e),s^n_d).
\end{align*}

\emph{Encoding. } To transmit message $m$ with the observed lossy description index $l_e$, the encoder finds the lossy description $v^n(l_e)$ and selects codeword $x^n(l_e,m)$.

\emph{Decoding.} Given $y^n\in\mathcal{Y}^n, s^n_d\in\mathcal{S}^n_d$ and $v^n$, the decoder looks for a unique message $\hat{m}$ such that
\begin{align*}
  (v^n(l_e),a^n(l_e,\hat{m}),x^n(l_e,\hat{m}),s^n_d,y^n)\in\mathcal{T}_1.
\end{align*}
If there is no such unique message $\hat{m}$ declare an error.

Without loss of generality, suppose message $m=1$ is sent. Define error events as follows.
\begin{align*}
  \mathcal{E}_1 = \{(V^n,S^n_d) \notin \mathcal{B} \;\text{for all $V^n\in\mathcal{C}_e$}\},\\
  \mathcal{E}_2 = \{(V^n(l_e),A^n(l_e,1),X^n(l_e,1),S^n_d,Y^n) \notin  \mathcal{T}_1\},\\
  \mathcal{E}_3 = \{(V^n(l_e),A^n(l_e,\hat{m}),X^n(l_e,\hat{m}),S^n_d,Y^n) \in  \mathcal{T}_1\;\text{for some $\hat{m}\neq 1$}\}
\end{align*}
The decoding error is bounded by $Pr\{\mathcal{E}\}\leq Pr\{\mathcal{E}_1\} + Pr\{\mathcal{E}_2 \cap \mathcal{E}_1^c\}+Pr\{\mathcal{E}_3\}$. %The bounding of the error probability is the same as the previous sections and is omitted here.
%\begin{comment}
To bound the first term on the right-hand side of the inequality, we omit the index for simplicity and it follows that
\begin{align*}
  Pr\{\mathcal{E}_1\} &=\sum_{a^n}P_{A^n}(a^n)\sum_{s^n_d}P_{S^n|A^n}(s^n_d|a^n)\left( \sum_{v^n}P_{V^n}(v^n)\mathbb{I}\{((v^n,s^n_d) \notin \mathcal{T}_2 \cup (v^n,a^n) \notin \mathcal{T}_3)\} \right)^{|\mathcal{C}_e|}\\
  &=\sum_{a^n}P_{A^n}(a^n)\sum_{s^n_d}P_{S^n|A^n}(s^n_d|a^n)\left(1 - \sum_{v^n}P_{V^n}(v^n)\mathbb{I}\{((v^n,s^n_d) \in \mathcal{T}_2 \cap (v^n,a^n) \in \mathcal{T}_3)\} \right)^{|\mathcal{C}_e|}\\
  &\leq \sum_{a^n}P_{A^n}(a^n)\sum_{s^n_d}P_{S^n|A^n}(s^n_d|a^n)\left(1 - 2^{-n(\bar{I}(\boldsymbol{V};\boldsymbol{S}_d) + \gamma)}\sum_{v^n}P_{V^n|S^n_d}(v^n|s^n_d)\mathbb{I}\{((v^n,s^n_d) \in \mathcal{T}_2 \cap (v^n,a^n) \in \mathcal{T}_3)\} \right)^{|\mathcal{C}_e|}\\
  &\leq \sum_{a^n}P_{A^n}(a^n)\sum_{s^n_d}P_{S^n|A^n}(s^n_d|a^n) \left(1 + \exp(-2^{n\gamma}) - \sum_{v^n}P_{V^n|S^n_d}(v^n|s^n_d)\mathbb{I}\{((v^n,s^n_d) \in \mathcal{T}_2 \cap (v^n,a^n) \in \mathcal{T}_3)\} \right)\\
  &= 1 + \exp(-2^{n\gamma}) - \sum_{a^n,s^n_d,v^n}P_{A^nS^n_dV^n}(a^n,s^n_d,v^n)\mathbb{I}\{((v^n,s^n_d) \in \mathcal{T}_2 \cap (v^n,a^n) \in \mathcal{T}_3)\}\\
  &\leq Pr\{(V^n,S^n_d)\notin \mathcal{T}_2\} + Pr\{(A^n,V^n)\notin \mathcal{T}_3\} + \exp(-2^{n\gamma}) \leq \pi_2 + \pi_3 + \exp(-2^{n\gamma})
\end{align*}
with $\pi_2\to 0,\;\pi_3 \to 0$ as $n\to \infty$.
The bounds of $\mathcal{E}_2$ and $\mathcal{E}_3$ are similar to the original general channel\cite[Chapter 3]{koga2013information} with $V^n$ being the common information at both encoder and decoder and are omitted here.
%\end{comment}
The average distortion is 
\begin{align*}
  &\mathbb{E}\left[ d_n(S^n,g(f_A(M),f(M,V^n),Z^n))\right] \\
  &=\frac{1}{|\mathcal{M}|}\sum_{m}\sum_{\mathcal{A}}P_{\textbf{A}}(\mathcal{A})\sum_{s^n_d}P_{S^n_d|A^n}(s^n_d|f_A(m))\sum_{\mathcal{C}_e}P_{\mathbf{C}_e}(\mathcal{C}_e)\sum_{x^n}P_{X^n|A^nV^n}(x^n|f_A(m),F^{\mathcal{C}_e}_n(s^n_d))\sum_{s^n}P_{S^n|A^nS^n_d}(s^n|f_A(m),s^n_d) \\
  &\quad\quad\quad\quad\quad\quad\quad\sum_{y^n}P_{Y^nZ^n|X^nS^n}(y^n,z^n|x^n,s^n)d_n(s^n,g(f_A(m),x^n,F^{\mathcal{C}_e}_n(s^n_d),z^n))\\
  &=\frac{1}{|\mathcal{M}|}\sum_{m}\sum_{a^n}P_{A^n}(a^n)\sum_{s^n_d}P_{S^n_d|A^n}(s^n_d|a^n)\sum_{\mathcal{C}_e}P_{\mathbf{C}_e}(\mathcal{C}_e)\sum_{x^n}P_{X^n|A^nV^n}(x^n|a^n,F^{\mathcal{C}_e}_n(s^n_d))\sum_{s^n}P_{S^n|A^nS^n_d}(s^n|a^n,s^n_d) \\
  &\quad\quad\quad\quad\quad\quad\quad\sum_{y^n}P_{Y^nZ^n|X^nS^n}(y^n,z^n|x^n,s^n)d_n(s^n,g(a^n,x^n,F^{\mathcal{C}_e}_n(s^n_d),z^n))\\
  &=\frac{1}{|\mathcal{M}|}\sum_{m}\sum_{s^n_d}P_{S^n_d}(s^n_d)\sum_{\mathcal{C}_e}P_{\mathbf{C}_e}(\mathcal{C}_e)\sum_{a^n}P_{A^n|S^n_d}(a^n|s^n_d)\sum_{x^n}P_{X^n|A^nV^n}(x^n|a^n,F^{\mathcal{C}_e}_n(s^n_d))\sum_{s^n}P_{S^n|A^nS^n_d}(s^n|a^n,s^n_d) \\
  &\quad\quad\quad\quad\quad\quad\quad\sum_{y^n}P_{Y^nZ^n|X^nS^n}(y^n,z^n|x^n,s^n)d_n(s^n,g(a^n,x^n,F^{\mathcal{C}_e}_n(s^n_d),z^n))\\
  &=\frac{1}{|\mathcal{M}|}\sum_{m}\sum_{s^n_d}P_{S^n_d}(s^n_d)\sum_{\mathcal{C}_e}P_{\mathbf{C}_e}(\mathcal{C}_e)\eta_2(s^n_d,F^{\mathcal{C}_e}_n(s^n_d)).
\end{align*}
The remaining proof is the same as that in Section \ref{sec: proof of average distortion} and is omitted here. The achievability proof is completed.

The converse part of $R_e$ directly follows by the converse part of \cite[Theorem 5.4.1]{koga2013information}. The bound on $R$ is the same as Section \ref{sec: proof of average distortion} by the definition of the code that we assign each message an action sequence and the fact that given $V^n$ the input sequence $X^n$ is also determined by the message. Here $V^n$ is the side information that is recoverable at both the encoder and decoder. By the conditional version of the converse part of \cite[Theorem 3.2.1]{koga2013information}, the proof is completed. 

The proof of the capacity under maximal distortion constraint is almost the same as the average distortion constraint above, except the distortion constraint part is replaced with a similar argument as that in Section \ref{sec: proof of maximal distortion}. We omit the detail here.

\bibliographystyle{ieeetr} 
\bibliography{ref}
\end{document}